%% file: templateArxiv.tex
\documentclass[twoside]{article}

\usepackage{PRIMEarxiv}
\usepackage{arxivpackages}
\usepackage{mymacros}
\usepackage{commands}
\renewcommand*{\backrefalt}[4]{%
    \ifcase #1 \footnotesize{(Not cited.)}%
    \or        \footnotesize{(Cited on page~#2.)}%
    \else      \footnotesize{(Cited on pages~#2.)}%
    \fi}

\renewcommand{\cite}[1]{\citep{#1}}
\setcitestyle{authoryear,round,citesep={;},aysep={,},yysep={;}}

\newcommand{\mytitle}[1][]{}
\newcommand{\mysubtitle}[1][]{}
\newcommand{\myauthors}[1][]{Denizalp Goktas}
\newcommand{\myauthorsshort}[1][]{Goktas}

\title{T\^atonnement in Homothetic Fisher Markets}

\author{
 Denizalp Goktas\\
  Department of Computer Science\\
  Brown University \\
  Providence, RI\\
  \texttt{ denizalp\_goktas@brown.edu} \\
   \And
Jiayi Zhao\\
 Department of Computer Science\\
 Pomona College \\
  Claremont, CA\\
  \texttt{jzae2019\@mymail.pomona.edu} \\
  \And
  Amy Greenwald \\
  Department of Computer Science \\
    Brown University \\
  Providence, RI \\
   \texttt{ amy\_greenwald@brown.edu} \\
}

\begin{document}
\maketitle
\input{abstract}

\input{intro}
\input{pseudo-gt}
\input{contributions}
\input{related}

\input{findings}
\input{prelim}
\input{homothetic}
\input{tatonnement}
\input{conclusion}

\section*{Acknowledgments}
We would like to thank Richard Cole, Yun Kuen Cheung, and Yixin Tao for very carefully reading an earlier version of this paper and providing invaluable feedback. Their efforts helped us refine the results in the paper.
Denizalp Goktas was supported by a JP Morgan AI Fellowship. 

\bibliographystyle{plainnat}  
\bibliography{references}  

\newpage
\appendix
\input{appendix/proofs}

\end{document}

%% file: abstract.tex
\begin{abstract}

A prevalent theme in the economics and computation literature is to identify natural price-adjustment processes by which sellers and buyers in a market can discover equilibrium prices.
An example of such a process is \emph{t\^atonnement}, an auction-like algorithm first proposed in 1874 by French economist Walras in which sellers adjust prices based on the Marshallian demands of buyers, i.e., budget-constrained utility-maximizing demands. 
A dual concept in consumer theory is a buyer's Hicksian demand, i.e., consumptions that minimize expenditure while achieving a desired utility level.
In this paper, we identify the maximum of the absolute value of the elasticity of the Hicksian demand, i.e., the maximum percentage change in the Hicksian demand of any good w.r.t.\@ the change in the price of some other good, as an economic parameter sufficient to capture and explain a range of convergent and non-convergent \emph{t\^atonnement\/} behaviors in a broad class of markets.
In particular, we prove the convergence of \emph{t\^atonnement\/} at a rate of $O(\nicefrac{(1+ \elastic^2)}{T})$, in homothetic Fisher markets with bounded price elasticity of Hicksian demand, i.e., Fisher markets in which consumers have preferences represented by homogeneous utility functions and the price elasticity of their Hicksian demand is bounded, where $\elastic$ is the maximum absolute value of the price elasticity of Hicksian demand across all buyers.
Our result not only generalizes known convergence results for CES Fisher markets, but extends them to mixed nested CES markets and Fisher markets with continuous, possibly non-concave, homogeneous utility functions.
Our convergence rate covers the full spectrum of nested CES utilities, including Leontief and linear utilities, unifying previously existing disparate convergence and non-convergence results.
In particular, for $\elastic = 0$, i.e., Leontief markets, we recover the best-known convergence rate of $O(\nicefrac{1}{T})$, and as $\elastic \to \infty$, e.g., linear Fisher markets, we obtain non-convergent behavior, as expected.

\keywords{Market Equilibrium \and Market Dynamics \and Fisher Markets.}

\end{abstract} 



%% file: intro.tex
\section{Introduction}
\label{sec:intro}


\mydef{Competitive (or Walrasian or market) equilibrium} \cite{arrow-debreu,walras}, first studied by French economist L\'eon Walras in 1874, is the steady state of an economy---any 
system governed by supply and demand \cite{walras}.
Walras assumed that each producer in an economy would act so as to maximize its profit, while consumers would make decisions that maximize their preferences over their available consumption choices; all this, while perfect competition prevails, meaning producers and consumers are unable to influence the prices that emerge.
Under these assumptions, the demand and supply of each commodity is a function of prices, as they are a consequence of the decisions made by the producers and consumers, having observed the prevailing prices.
A competitive equilibrium then corresponds to prices that solve the system of simultaneous equations with demand on one side and supply on the other, i.e., prices at which supply meets demand.
Unfortunately, Walras did not provide conditions that guarantee the existence of such a solution, and the question of whether such prices exist remained open until \citeauthor{arrow-debreu}'s rigorous analysis of competitive equilibrium in their model of a competitive economy in the middle of last century \cite{arrow-debreu}.

The Arrow-Debreu model comprises a set of commodities; a set of firms, each deciding what quantity of each commodity to supply; and a set of consumers, each choosing a quantity of each commodity to demand in exchange for their endowment \cite{arrow-debreu}. 
\citeauthor{arrow-debreu} define a \mydef{competitive equilibrium} as a collection of consumptions, one per consumer, a collection of productions, one per firm, and prices, one per commodity, such that fixing equilibrium prices: (1) no consumer can increase their utility by deviating to an alternative affordable consumption, (2) no firm can increase profit by deviating to another production in their production set, and (3) the \mydef{aggregate demand} for each commodity (i.e., the sum of the commodity's consumption across all consumers) does not exceed to its \mydef{aggregate supply} (i.e., the sum of the commodity's production and endowment across firms and consumers, respectively), while the total value of the aggregate demand is equal to the total value of the aggregate supply, i.e., \mydef{Walras' law} holds.

\citeauthor{arrow-debreu} proceeded to show that their competitive economy could be seen as an \mydef{abstract economy}, which today is better known as a \mydef{pseudo-game} \cite{arrow-debreu, facchinei2010generalized}.
A pseudo-game is a generalization of a game in which the actions taken by each player impact not only the other players' payoffs, as in games, but also their set of permissible actions.
\citeauthor{arrow-debreu} proposed \mydef{generalized Nash equilibrium} as the solution concept for this model, an action profile from which no player can improve their payoff by unilaterally deviating to another action in the space of permissible actions determined by the actions of other players.
\citeauthor{arrow-debreu} further showed that any competitive economy could be represented as a pseudo-game inhabited by a fictional auctioneer, who sets prices so as to buy and resell commodities at a profit, as well as consumers and producers, who respectively, choose utility-maximizing consumptions of commodities in the budget sets determined by the prices set by the auctioneer, and profit-maximizing productions at the prices set by the auctioneer.
The elegance of the reduction from competitive economies to pseudo-games is rooted in a simple observation: the set of competitive equilibria of a competitive economy is equal to the set of generalized Nash equilibria of the associated pseudo-game, implying the existence of competitive equilibrium in competitive economies as a corollary of the existence of generalized Nash equilibria in pseudo-games, whose proof is a straightforward generalization of Nash's proof for the existence of Nash equilibria \cite{nash1950existence}.%
\footnote{\citet{mckenzie1959existence} would prove the existence of competitive equilibrium independently, but concurrently.
Much of his work, however, has gone unrecognized perhaps because his proof technique does not depend on this fundamental relationship between competitive and abstract economies.}

With the question of existence out of the way, this line of work on competitive equilibrium, which today is known as general equilibrium theory \cite{mckenzie2005classical}, turned its attention to questions of (1) efficiency, (under what assumptions are competitive equilibria Pareto-optimal?) (2) uniqueness (under what assumptions are competitive equilibria unique?), and (3) stability (under what conditions would a competitive economy settle into a competitive equilibrium?).
The first two questions were answered between the 1950s and 1970s \cite{arrow1951extension, arrow1958note, arrow-hurwicz, balasko1975some, debreu1951pareto, dierker1982unique, hahn1958gross, pearce193unique}, showing that (1) under suitable assumptions (e.g., see \citet{arrow-welfare}) competitive equilibrium demands are Pareto-optimal, and (2) competitive equilibria are unique in markets with an \mydef{excess demand function}, (i.e., the difference between the aggregate demand and supply functions), which satisfies the \mydef{gross substitutes} (GS) condition (i.e., the excess demand of any commodity increases if the price of any other commodity increases, fixing all other prices).
In regards to the question of stability, most relevant work is concerned with the convergence properties of a natural auction-like price-adjustment process, known as \mydef{t\^atonnement}, which mimics the behavior of the \mydef{law of supply and demand}, updating prices at a rate equal to the excess demand \cite{arrow1971general, kaldor1934classificatory}.
Research on \emph{t\^atonnement\/} in the economics literature is motivated by the fact that it can be understood as a plausible explanation of how prices move in real-world markets.
Hence, if one could prove convergence in all exchange economies, then perhaps it would be justifiable to claim real-world markets would also eventually settle at a competitive equilibrium.

\citet{walras}
conjectured, albeit without conclusive evidence, that \emph{t\^atonnement\/} would converge to a competitive equilibrium.
While a handful of results guarantee the convergence of \emph{t\^atonnement\/} under mathematical conditions without widely agreed-upon economic interpretations \cite{nikaido1960stability, uzawa1960walras}, \citeauthor{arrow-hurwicz} [1958; 1960]
were the first to formally establish the convergence of \emph{t\^atonnement} to unique competitive equilibrium prices in a class of economically well-motivated competitive economies, namely those that satisfy the GS assumption.
Following this promising result, \citet{scarf1960instable} dashed all hope that \emph{t\^atonnement\/} would prove to be a universal price-adjustment process that converges in all economies, by showing that competitive equilibrium prices are unstable under \emph{t\^atonnement\/} dynamics in his eponymous competitive economy without firms, and with only three commodities and three consumers with Leontief preferences, i.e., \mydef{the Scarf exchange economy}.
Scarf's negative result seems to have discouraged further research by economists on the stability of competitive equilibrium \cite{fisher1975stability}.
Despite research on this question coming to a near halt, one positive outcome was achieved, on the convergence of a non-\emph{t\^atonnement\/} update rule known as \mydef{Smale's process} \cite{herings1997globally, kamiya1990globally, van1987convergent, smale1976convergent}, which updates prices at the rate of the product of the excess demand and the inverse of its Jacobian, in most competitive economies, even beyond GS, again suggesting the possibility that real-world economies could indeed settle at a competitive equilibrium. 

Nearly half a century after these seminal analyses of competitive economies, research on the stability of competitive equilibrium is once again coming to the fore, this time in computer science, perhaps motivated by applications of algorithms such as \emph{t\^atonnement\/} to load balancing over networks \cite{jain2013constrained}, or to pricing of transactions on crypotocurrency blockchains \cite{leonardos2021dynamical, liu2022empirical, reijsbergen2021transaction}.
A detailed inquiry into the computational properties of market equilibria was initiated by \citet{devanur2008market}, who studied a special case of the Arrow-Debreu competitive economy known as the \mydef{Fisher market} \cite{brainard2000compute}.
This model, for which Irving Fisher computed equilibrium prices using a hydraulic machine in the 1890s, is essentially the Arrow-Debreu model of a competitive economy, but there are no firms, and buyers are endowed with only one type of commodity---hereafter good%
\footnote{In the context of Fisher markets, commodities are typically referred to as goods \citep{fisher-tatonnement}, as Fisher markets are often analyzed for a single time period only.
More generally, in Arrow-Debreu markets, where commodities vary by time, location, or state of the world, "an apple today" may be different than "an apple tomorrow". For consistency with the literature, we refer to commodities as goods. }---an artificial currency 
\cite{brainard2000compute, AGT-book}.
\citet{devanur2002market} exploited a connection first made by \citet{eisenberg1961aggregation} between the \mydef{Eisenberg-Gale program} and competitive equilibrium to solve Fisher markets assuming buyers with linear utility functions, thereby providing a (centralized) polynomial-time algorithm for equilibrium computation in these markets~\cite{devanur2002market,devanur2008market}.
Their work was built upon by \citet{jain2005market}, who extended the Eisenberg-Gale program to all Fisher markets in which buyers have \mydef{continuous, quasi-concave, and homogeneous} utility functions, and proved that the equilibrium of Fisher markets with such buyers can be computed in polynomial time by interior point methods. 

Concurrent with this line of work on computing competitive equilibrium using centralized methods, a line of work on devising and proving 
convergence guarantees for decentralized price-adjustment processes (i.e., iterative algorithms that update prices according to a predetermined update rule) developed.
This literature has focused on devising \emph{natural\/} price-adjustment processes, like \emph{t\^atonnement}, which might explain or imitate the movement of prices in real-world markets.
In addition to imitating the law of supply and demand, \emph{t\^atonnement} has been observed to replicate the movement of prices in lab experiments, where participants are given endowments and asked to trade with one another \cite{gillen2020divergence}.
Perhaps more importantly, the main premise of research on the stability of competitive equilibrium in computer science 
is that for competitive equilibrium to be justified, not only should it be backed by a natural price-adjustment process as economists have long argued, but it should also be computationally efficient \cite{AGT-book}.

The first result on this question is due to \citet{codenotti2005market}, who introduced a discrete-time version of \emph{t\^atonnement}, 
and showed that in exchange economies that satisfy \mydef{weak gross substitutes (WGS)}, the \emph{t\^atonnement\/} process converges to an approximate competitive equilibrium in a number of steps which is polynomial in the approximation factor and size of the problem.
Unfortunately, soon after this positive result appeared, \citet{papadimitriou2010impossibility} showed that it is impossible for a price-adjustment process based on the excess demand function to converge in polynomial time to a competitive equilibrium in general, ruling out the possibility of Smale's process (and many others)
justifying the notion of competitive equilibrium in all competitive economies.
Nevertheless, further study of the convergence of price-adjustment processes such as \emph{t\^atonnement\/} under stronger assumptions, or in simpler models than full-blown Arrow-Debreu competitive economies, remains worthwhile, as these processes are being deployed in practice \cite{jain2013constrained, leonardos2021dynamical, liu2022empirical, reijsbergen2021transaction}.

Toward this end, in this paper we make strides towards analyzing the computational complexity of discrete-time \emph{t\^atonnement\/} in \mydef{homothetic Fisher markets}, i.e., Fisher markets in which consumers have continuous and homothetic preferences.%
\footnote{We refer to Fisher markets that comprise buyers with a certain utility function by the name of the utility function, e.g., we call a Fisher market that comprises buyers with Leontief utility functions a Leontief Fisher market. We omit the ``continuous'' qualifier as competitive equilibrium is not guaranteed to exist when preferences are not continuous.}
An important concept in consumer theory is a buyer's Hicksian demand, i.e., consumptions that minimize expenditure while achieving a desired utility level.
In this paper, we identify the maximum elasticity of the Hicksian demand, i.e., the maximum percentage change in the Hicksian of any good w.r.t.\@ the change in the price of some other good, as an economic parameter sufficient to capture and explain a range of convergent and non-convergent \emph{t\^atonnement\/} behaviors in a broad class of markets.
In particular, we prove the convergence of \emph{t\^atonnement\/} in homothetic Fisher markets with bounded elasticity of Hicksian demand, i.e., Fisher markets in which consumers have preferences represented by homogeneous utility functions for which the elasticity of their Hicksian demand is bounded.

%% file: pseudo-gt.tex
\paragraph{Interpretation via Pseudo-Game Theory}


Recall that in the pseudo-game associated with a competitive economy \cite{arrow-debreu}, a (fictional) auctioneer sets prices for commodities, firms choose what quantity of each commodity to produce, and consumers choose what quantity of each commodity to consume in exchange for their endowment.
Running \emph{t\^atonnement\/} in this pseudo-game amounts to the auctioneer running a first-order method, namely a gradient ascent dynamic on their profit function,
while the consumers and firms reply with their best response.
Interestingly, if the competitive economy satisfies WGS, then the excess demand function is monotone\footnote{\sdeni{}{Technically speaking, the excess demand is strictly pseudomonotone, but we ignore this distinction for simplicity.}} 
over all \emph{t\^atonnement\/} trajectories. 
This pseudo-game can then be understood as a monotone variational inequality \cite{facchinei2010generalized}, 
whose solutions are competitive equilibrium prices. 

As \emph{t\^atonnement\/} is intended to be an explanation of real-world market behavior, proofs of its convergence would ideally rely on justifiable economic assumptions. 
Via the connection between pseudo-games games and competitive economies, establishing convergence of \emph{t\^atonnement\/} 
can be reduced to discovering justifiable economic assumptions 
to impose so that the ensuing pseudo-game or variational inequality satisfies suitable mathematical conditions for convergence. 

Unfortunately, without imposing significant additional assumptions (such as the excess demand function of the competitive economy being Lipschitz-smooth; see, for example, \citet{golowich2020eglast}) or relying on more complex update rules such as \mydef{extragradient descent} \cite{korpelevich1976extragradient}, first-order methods are not guaranteed to converge in last iterates%
\footnote{We note that the standard convergence metric in the \emph{t\^atonnement\/} literature is convergence in last iterates (see, for example, \citet{fisher-tatonnement}).}
in monotone 
pseudo-games or monotone variational inequalities.
Undeterred, in addition to assuming WGS, \citet{cole2008fast} impose economic assumptions on the Marshallian own-price elasticity of demand and Marsallian income elasticity of demand, which imply Lipschitz-smoothness of the excess demand function over \emph{t\^atonnement\/} trajectories, and obtain convergence of \emph{t\^atonnement\/} in last iterates in Fisher markets with WGS and bounded price/income elasticity of Marshallian demand. 

Paralleling the duality between WGS competitive economies and monotone pseudo-games, any homothetic Fisher market is equivalent to a zero-sum game.
Moreover, this zero-sum game can be further reduced to a convex potential, i.e., the Eisenberg-Gale program's dual, whose solutions are competitive equilibrium prices \cite{devanur2002market, fisher-tatonnement}.
With this equivalence in hand, we seek to identify economically justifiable assumptions that translate into mathematical conditions on the excess demand function that are sufficient for convergence, 
because first-order methods run on zero-sum games and convex potentials are otherwise not guaranteed to converge to an optimal solution in last iterates.
One obvious candidate condition is Lipschitz-smoothness of the excess demand function, but this property does not hold in Leontief Fisher markets, a flavor of homothetic markets in which \emph{t\^atonnement\/} converges!
Our research has led to the discovery that assuming bounded elasticity of Hicksian demand yields an excess demand function that is Lipschitz-continuous as well as Bregman-smooth w.r.t.\@ to the KL divergence over trajectories of \emph{t\^atonnement},%
\footnote{Bregman-smoothness is a generalization of Lipschitz-smoothness introduced by \citet{cheung2018dynamics} following work by \citet{grad-prop-response}; see \Cref{sec:prelim} for the mathematical definition.} two properties which are sufficient to guarantee the convergence of \emph{t\^atonnement} in homothetic Fisher markets. 
\emph{Our contribution, then, is to identify the economic assumptions that imply the requisite mathematical properties that yield convergence of \emph{t\^atonnement\/} in last iterates.}


\begin{table}[]
\centering
\begin{tabular}{|c|c|c|}
\hline
Economy type & Pseudo-game type & Mathematical object      \\
\hline \hline
WGS Economy & Monotone pseudo-game & Monotone Variational Inequality \\
Homothetic Fisher market & Zero-sum game & Convex Potential \\
\hline
\end{tabular}
\caption{Summary of the equivalences among economy types, pseudo-game types, and mathematical objects.}
\end{table}

%% file: contributions.tex
\subsection{Technical Contributions}

Earlier work \cite{cheung2014analyzing, fisher-tatonnement} has established a convergence rate of $\left(1- \Theta(1)\right)^T$ for CES Fisher markets excluding the linear and Leontief cases, and of $O\left(\nicefrac{1}{T}\right)$ for Leontief 
and nested%
\footnote{See Chapter 10 of \citet{cheung2014analyzing}.}
CES Fisher markets, where $T \in \N_+$ is the number of iterations for which \emph{t\^atonnement\/} is run.
In linear Fisher markets, however, \emph{t\^atonnement\/} does not converge.
We generalize these results by proving a convergence rate of $O(\nicefrac{(1+ \elastic^2)}{T})$, where $\elastic$ is the maximum absolute value of the price elasticity of Hicksian demand across all buyers. 
Our convergence rate covers the full spectrum of homothetic Fisher markets, including mixed CES markets, i.e., CES markets with linear, Leontief, and (nested) CES buyers, unifying previously existing disparate convergence and non-convergence results.
In particular, for $\elastic = 0$, i.e., Leontief Fisher markets, we recover the best-known convergence rate of $O(\nicefrac{1}{T})$, and as $\elastic \to \infty$, i.e., linear Fisher markets, we obtain the non-convergent behaviour of \emph{t\^atonnement} \cite{cole2019balancing}. We summarize known convergence results in light of our results in \Cref{fig:utility-functions}.

We observe that, in contrast to general competitive economies, in homothetic Fisher markets, concavity of the utility functions is not necessary for the existence of competitive equilibrium (\Cref{new-convex}).
A computational analog of this result also holds, namely that \emph{t\^atonnement\/} converges in homothetic Fisher markets, even when buyers' utility functions are non-concave. 
Our results parallel known results on the convergence of \emph{t\^atonnement\/} in WGS markets, where concavity of utility functions is again not necessary for convergence \cite{codenotti2005market}. 

\begin{figure}
\begin{subfigure}{0.45\textwidth}
\resizebox{\columnwidth}{!}
{
    \includegraphics[]{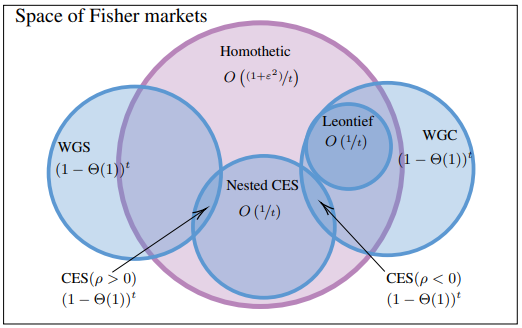}
}
\caption{
    The convergence rates of \emph{t\^atonnement} for different Fisher markets. We color previous contributions in blue, and our contribution in red, i.e., we study homothetic Fisher markets where $\elastic$ is the maximum absolute value of the price elasticity of Hicksian demand across all buyers. We note that the convergence rate for WGS markets does not apply to markets where the price elasticity of Marshallian demand is unbounded, e.g., linear Fisher markets; likewise, the convergence rate for nested CES  Fisher markets does not apply to linear or Leontief Fisher markets.}
    \label{fig:utility-functions}
\end{subfigure}
\quad \quad
\begin{subfigure}{0.45\textwidth}
 \resizebox{1.1\columnwidth}{!}
{
 \hspace{-2em}
 \includegraphics[]{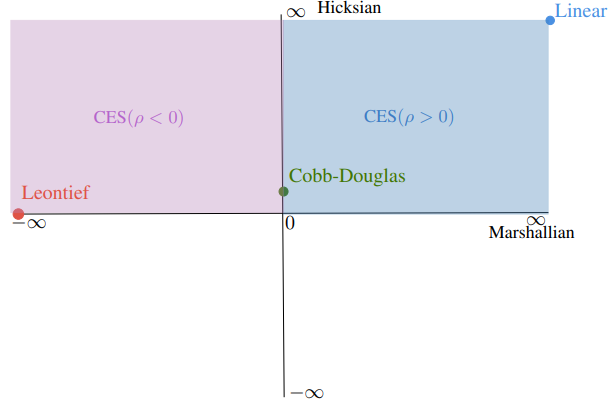}
} 
\caption{Cross-price elasticity taxonomy of well-known homogeneous utility functions. There are no previously studied utility functions in the space of utility functions with negative Hicksian cross-price elasticity. Future work could investigate this space and prove faster convergence rates than those provided in this paper. We note that our convergence result covers the entire spectrum of this taxonomy (excluding limits of the $y$-axis).}\label{fig:taxonomy}
\end{subfigure}
\caption{A summary of known results in Fisher markets.}
\end{figure}



%% file: related.tex
\subsection{Related Work}

Following \citeauthor{codenotti2005market}'s [\citeyear{codenotti2005market}] initial analysis of \emph{t\^atonnement\/} in competitive economies that satisfy WGS, \citet{garg2004auction} introduced an auction algorithm that also converges in polynomial time for linear exchange economies.
More recently, \citet{bei2015tatonnement} established faster convergence bounds for \emph{t\^atonnement\/} in WGS exchange economies.

Another line of work considers price-adjustment processes in variants of Fisher markets.
\citet{cole2008fast} analyzed \emph{t\^atonnement\/} in a real-world-like model satisfying WGS called the ongoing market model.
In this model, \emph{t\^atonnement\/} once-again converges in polynomial-time \cite{cole2008fast, cole2010discrete}, and it has the advantage that it can be seen as an abstraction for market processes.
\citeauthor{cole2008fast}'s results were later extended by \citet{cheung2012tatonnement} to ongoing markets with \mydef{weak gross complements}, i.e., the excess demand of any commodity weakly increases if the price of any other commodity weakly decreases, fixing all other prices, and ongoing markets with a mix of WGC and WGS commodities.
The ongoing market model these two papers study contains as a special case the Fisher market; however \citet{cole2008fast} assume bounded own-price elasticity of Marshallian demand, and bounded income elasticity of Marshallian demand, while \citet{cheung2012tatonnement} assume, in addition to \citeauthor{cole2008fast}'s assumptions, bounded adversarial market elasticity, which can be seen as a variant of bounded cross-price elasticity of Marshallian demand, from below.
With these assumptions, these results cover Fisher markets with a small range of the well-known CES utilities, including CES Fisher markets with $\rho \in [0, 1)$ and WGC Fisher markets with $\rho \in (- 1, 0]$.%
\footnote{We refer the reader to \Cref{sec:prelim} for a definition of CES utilities in terms of the substitution parameter $\rho$.}

\citet{fisher-tatonnement} built on this work by establishing the convergence of \emph{t\^atonnement\/} in polynomial time in nested CES Fisher markets, excluding the limiting cases of linear and Leontief markets, but nonetheless extending polynomial-time convergence guarantees for \emph{t\^atonnement\/} to Leontief Fisher markets as well.
More recently, \citet{cheung2018amortized} showed that \citeauthor{fisher-tatonnement}'s [\citeyear{fisher-tatonnement}] result extends to an asynchronous version of \emph{t\^atonnement}, in which good prices are updated during different time periods. 
In a similar vein, \citet{cheung2019tracing} analyzed \emph{t\^atonnement\/} in online Fisher markets, determining that \emph{t\^atonnement\/} tracks competitive equilibrium prices closely provided the market changes slowly.

Another price-adjustment process that has been shown to converge to market equilibria in Fisher markets is \mydef{proportional response dynamics}, first introduced by \citet{first-prop-response} for linear utilities; then expanded upon and shown to converge by \cite{proportional-response} for all CES utilities; and very recently shown to converge in Arrow-Debreu exchange economies with linear and CES ($\rho \in [0,1)$) utilities by \citeauthor{branzei2021proportional}. 
The study of the proportional response process was proven fundamental when \citeauthor{fisher-tatonnement} noticed its relationship to gradient descent.
This discovery opened up a new realm of possibilities in analyzing the convergence of market equilibrium processes.
For example, it allowed \citet{cheung2018dynamics} to generalize the convergence results of proportional response dynamics to Fisher markets for buyers with mixed CES utilities.
This same idea was applied by \citet{fisher-tatonnement} to prove the convergence of \emph{t\^atonnement\/} in Leontief Fisher markets, using the equivalence between mirror descent \cite{boyd2004convex}
on the dual of the Eisenberg-Gale program 
and \emph{t\^atonnement}, first observed by \citet{devanur2008market}.
More recently, \citet{gao2020first} developed 
methods to solve the Eisenberg-Gale convex program in the case of linear, quasi-linear, and Leontief Fisher markets.

An alternative to the (global) competitive economy model, in which an agent's trading partners are unconstrained, is the \citet{kakade2004graphical} model of a graphical economies.
This model features local markets, in which each agent can set its own prices for purchase only by neighboring agents, and likewise can purchase only from neighboring agents. 
Auction-like price-adjustment processes have been shown to converge in variants of this model assuming WGS \cite{andrade2021graphical}.

%% file: prelim.tex
\section{Preliminaries}
\label{sec:prelim}
\paragraph{Notation.}
We use caligraphic uppercase letters to denote sets (e.g., $\calX$); bold lowercase letters to denote vectors (e.g., $\price, \bm \pi$);
bold uppercase letters to denote matrices 
{(e.g., $\allocation$) and}
lowercase letters to denote scalar quantities (e.g., $x, \delta$).
We denote the $i$th row vector of a matrix (e.g., $\allocation$) by the corresponding bold lowercase letter with subscript $i$ (e.g., $\allocation[\buyer])$. 
Similarly, we denote the $j$th entry of a vector (e.g., $\price$ or $\allocation[\buyer]$) by the corresponding Roman lowercase letter with subscript $j$ (e.g., $\price[\good]$ or $\allocation[\buyer][\good]$).
%
If a correspondence $\calA: \calX \rightrightarrows \calY$ is singleton valued for some $\x \in \calX$, for notational convenience, we treat it as an element of $\calX$, i.e., $\calA(\x) \in \calX$, rather than a subset of it, i.e., $\calA(\x) \subset \calX$.
We denote the set of numbers $\left\{1, \hdots, n\right\}$ by $[n]$, the set of natural numbers by $\N$, the set of real numbers by $\R$, the set of non-negative real numbers by $\R_+$ and the set of strictly positive real numbers by $\R_{++}$.  
We let $\simplex[n] = \{\x \in \R_+^n \mid \sum_{i = 1}^n x_i = 1\}$. We denote the interior of any set $\calA$ by $\mathrm{int}(A)$. We define $\ball[\varepsilon][\x] = \{ \z \in \calZ \mid ||\z - \x || \leq \varepsilon \}$ to be the closed $\varepsilon$-ball centered at $\x$, where $\calZ$ and $\|\cdot \|$ will be clear from context.
%
%
We denote the partial derivative of a function $f: X \to \R$ for $X \subset \R^\numbuyers$ w.r.t. $x_i$ at a point $\x = \y$ by $\subdiff[x_i] f(\y)$.
We define the gradient $\grad[\x]: C^1(\calX) \to C^0(\calX) $ as the operator which takes as input a functional $f: \calX \to \R$, and outputs a vector-valued function consisting of the partial derivatives of $f$ w.r.t. $\x$. Finally, by notational overload, we define the subdifferential $\subgrad$ of a function $f$ at a point $\bm{a} \in U$ by $\subdiff[\bm{x}] f(\bm{a}) = \{\subgrad \mid f(\bm{x}) \geq f(\bm{a}) + \subgrad^T (\bm{x} - \bm{a}) \}$.


\subsection{Mathematical Preliminaries}
%
%
%
%
%
A function $f: \R^m \to \R$ is said to be \mydef{homogeneous of degree $k \in \N_+$} if $\forall \allocation[ ] \in \R^m, \lambda > 0, f(\lambda \allocation[ ]) = \lambda^k f(\allocation[ ])$.
Unless otherwise indicated, without loss of generality, a homogeneous function is assumed to be homogeneous of degree 1.
Fix any norm $\| \cdot \|$. Given $\calA \subset \R^d$, the function $\obj: \calA \to \R$ is said to be $\lipschitz[\obj]$-\mydef{Lipschitz-continuous} iff $\forall \x_1, \x_2 \in \calX, \left\| \obj(\x_1) - \obj(\x_2) \right\| \leq \lipschitz[\obj] \left\| \x_1 - \x_2 \right\|$.
If the gradient of $\obj$ is $\lipschitz[\grad \obj]$-Lipschitz-continuous, we then refer to $\obj$ as $\lipschitz[\grad \obj]$-\mydef{Lipschitz-smooth}.%
\subsection{Mirror Descent}

Consider the optimization problem $\min_{\bm{x} \in V} f(\bm{x})\label{optimization-problem}$, where $f: \R^n \to \R$ is a differentiable convex function 
and $V$ is the feasible set of solutions.
%
A standard method for solving this problem is the \mydef{mirror descent algorithm} \cite{boyd2004convex}:
\begin{align}
    \bm{x}(t+1) &= \argmin_{\bm{x} \in V} \left\{ \lapprox[f][\x][\bm{x}(t)] + \gamma_t \divergence[h][\bm{x}][\bm{x}(t)] \right\} && \text{for }t = 0, 1, 2, \ldots \label{generalized-descent} \\
    \bm{x}(0) &\in \mathbb{R}^{n}\label{generalized-descent2}
\end{align}

\noindent
Here, $\gamma_t > 0$ is the step size at time $t$, $\lapprox[f][\bm{x}][\bm{y}]$ is the \mydef{linear approximation} of $f$ at $\bm{y}$, that is $\lapprox[f][\bm{x}][\y] = f(\y) + \grad f(\y)^T (\x - \y)$, and $\divergence[h][\bm{x}][\bm{x}(t)]$ is the \mydef{Bregman divergence} of a convex differentiable \mydef{kernel} function $h(\bm{x})$ defined as $\divergence[h][\bm{x}][\bm{y}] = h(\bm{x}) - \lapprox[h][\bm{x}][\bm{y}]$ \cite{bregman1967relaxation}. 
In particular, when $h(\bm{x})= \frac{1}{2}||\bm{x}||_2^2$, $\divergence[h][\bm{x}][\bm{y}] = \frac{1}{2} ||\bm{x} - \bm{y}||_2^2$.
In this case, mirror descent reduces to projected gradient descent \cite{boyd2004convex}. 
If instead the kernel is the 
weighted entropy $h(\x) =  \sum_{i \in [n]} \left(x_{i} \log(x_{i}) - x_i \right)$, 
the Bregman divergence reduces to the \mydef{generalized Kullback-Leibler (KL) divergence} \cite{joyce2011kullback}:
\begin{align}
    \divergence[\mathrm{KL}][\x][\y] = 
    \sum_{i \in [n]} \left[ x_i \log\left( \frac{x_i}{y_i}\right) - x_i +  y_i \right]
    \enspace ,
\end{align}

\noindent
which, when $V = \R^m_{++}$, yields the following simplified
\mydef{entropic descent} update rule:
\begin{align}
    & \forall \good \in \goods & x_{\good}^{(t+1)} &= x_{\good}^{(t)} \exp \left\{\frac{-\subdiff[x_\good]f(\x^{(t)})}{\gamma_t} \right\} && \text{for }t = 0, 1, 2, \ldots \label{exp:subgradient-descent} \\
    & & x_{\good}^{(0)} &\in \R_{++}
    \label{exp:subgradient-descent-init}
\end{align}

A function $f$ is said to be \mydef{$\gamma$-Bregman-smooth} \cite{cheung2018dynamics} w.r.t.\ a Bregman divergence with kernel function $h$ if $f(\bm{x}) \leq \lapprox[f][\bm{x}][\bm{y}] + \gamma \divergence[h][\bm{x}][\bm{y}]$.
\citet{grad-prop-response} showed that if the objective function $f(\bm{x})$ of a convex optimization problem is \mydef{$\gamma$-Bregman} w.r.t.\@ to some Bregman divergence $\divergence[h]$, then mirror descent with Bregman divergence $\divergence[h]$ converges to an optimal solution $f(\bm{x}^*)$ at
a rate of $O \mleft( \nicefrac{1}{t} \mright)$.
We require a slightly modified version of this theorem, introduced by \citet{fisher-tatonnement}, where it suffices for the $\gamma$-Bregman-smoothness property to hold only for consecutive pairs of iterates.

\begin{theorem}[\citet{grad-prop-response},\citet{fisher-tatonnement}]
\label{theorem-devanur}
Let $\{\x^{t}\}_{t}$ be the iterates generated by mirror descent with Bregman divergence $\divergence[h]$. Suppose $f$ and $h$ are convex, and for all $t \in \N$ and for some $\gamma > 0$, it holds that
$f(\bm{x}^{(t+1)}) \leq \lapprox[f][\bm{x}^{(t+1)}][\bm{x}^{(t)}] + \gamma \divergence[h][\bm{x}^{(t+1)}][\bm{x}^{(t)}]$. 
If $\bm{x}^*$ is a minimizer of $f$, then the following holds for mirror descent
with fixed step size $\gamma$: for all $t \in \N$,
$f(\bm{x}^{(t)}) - f(\bm{x}^*) \leq \nicefrac{\gamma}{t} \, \divergence[h][\bm{x}^*][\bm{x}^{(0)}]$.
\end{theorem}

\subsection{Consumer Theory}

Let $\choiceset = \R_+^\numgoods$ be a set of possible consumptions over $\numgoods$ goods s.t.\@ for any $\allocation[ ] \in \R^\numgoods_+$ and $\good \in \goods$, $\allocation[ ][\good] \geq 0$ represents the amount of good $\good \in \goods$ consumed by consumer (hereafter, buyer) $\buyer$. 
The preferences of buyer $\buyer$ over different consumptions of goods can be represented by a \mydef{preference relation} $\succeq_\buyer$ over $\choiceset$ such that the buyer (resp.\@ weakly) prefers a choice $\x \in \choiceset$ to another choice $\y \in \choiceset$ iff $\x \succ_\buyer \y$ (resp.\@ $\x \succeq_\buyer \y$). 
A preference relation is said to be \mydef{complete} iff for all $\x, \y \in \choiceset$, either $\x \succeq_\buyer \y$ or $\y \succeq_\buyer \x$, or both.
A preference relation is said to be \mydef{transitive} if, for all $\x, \y, \z \in \choiceset$, $\x \succeq_\buyer \z$ whenever $\x \succeq_\buyer \y$ and $\y \succeq_\buyer \z$.
A preference relation is said to be \mydef{continuous} if for any sequence $\{ \x^{(n)}, \y^{(n)}\}_{n \in \N_+} \subset \choiceset \times \choiceset$ $(\x^{(n)}, \y^{(n)}) \to (\x, \y)$ and $\x^{(n)} \succeq_\buyer \y^{(n)}$ for all $n \in \N_+$, it also holds that $\x \succeq_\buyer \y$. 
A preference relation $\succeq_\buyer$ is said to be \mydef{locally non-satiated} iff for all $\x \in \choiceset$ and $\epsilon > 0$, there exists $\y \in \ball[\epsilon][\x]$ such that $\y \succ_\buyer \x$. 
A utility function $\util[\buyer]: \choiceset \to \R_+$ assigns a positive real value%
\footnote{Without loss of generality, we assume that utility functions are positive real-valued functions, since any real-valued function can be made positive real-valued by passing it through the monotonic transformation $\x \mapsto e^\x$ without affecting the underlying preference relation.}
to elements of $\choiceset$, i.e., to all possible consumptions.
Every continuous utility function represents some complete, transitive, and continuous preference relation $\succeq_\buyer$ over goods s.t.\@ if $\util[\buyer](\x) \geq \util[\buyer](\y)$ for two bundles of goods $\x, \y \in \R^\numgoods$, then $\x \succeq_\buyer \y$ \cite{arrow1971general}. 

In this paper, we consider the general class of \mydef{homothetic} preferences $\succeq_\buyer$ s.t.\@ for any consumption $\x, \y \in \X$ and $\lambda \in \R_+$, $\x \succeq_\buyer \y$ and $\y \succeq_\buyer \x$ implies $\lambda \x \succeq_\buyer \lambda \y$ and $\lambda \y\succeq_\buyer \lambda \x$, respectively. 
A preference relation $\succeq_\buyer$ is complete, transitive, continuous, and homothetic iff it can be represented via a continuous and homogeneous utility function $\util[\buyer]$ of arbitrary degree \cite{arrow1971general}.%
\footnote{Throughout this work, without loss of generality, we assume that complete, transitive, continuous, and homothetic preference relations are represented via a homogeneous utility function of degree 1, since any homogeneous utility function of degree $k$ can be made homogeneous of degree 1 without affecting the underlying preference relation by passing the utility function through the monotonic transformation $x \mapsto \sqrt[k]{x}$.} 
We note that any homogeneous utility function $\util[\buyer]$ represents locally non-satiated preferences, since for all $\epsilon > 0$ and $\x \in \choiceset$, there exists an allocation $(1+\nicefrac{\varepsilon}{\left\| \x \right\|})\x$ s.t.\@ $\util[\buyer]((1+\nicefrac{\varepsilon}{\left\| \x \right\|})\x) = (1+\nicefrac{\varepsilon}{\left\| \x \right\|})\util[\buyer](\x)> \util[\buyer](\x)$, and $[\x - (1+\nicefrac{\varepsilon}{\left\| \x \right\|})\x] \in \ball[\varepsilon][\x]$. 

The class of homogeneous utility functions includes the well-known \mydef{constant elasticity of substitution (CES)} utility function family, parameterized by a substitution parameter $-\infty \leq \rho_\buyer \leq 1$, and
given by
$\util[\buyer](\allocation[\buyer]) = \sqrt[{\rho_\buyer}]{ \sum_{\good \in \goods} \valuation[\buyer][\good] \allocation[\buyer][\good]^{\rho_\buyer}}$ with each utility function parameterized by the vector of valuations $\valuation[\buyer] \in \mathbb{R}_+^{\numbuyers}$, where each $\valuation[\buyer][\good]$ quantifies the value of good $j$ to buyer $i$. CES utilities are said to be \mydef{gross substitutes} (resp. \mydef{gross complements}) CES if $\rho_\buyer > 0$ ($\rho_\buyer < 0$).
\mydef{Linear utility} functions are obtained when $\rho$ is $1$ (goods are perfect substitutes), while \mydef{Cobb-Douglas} and \mydef{Leontief utility} functions are obtained when $\rho \to 0$ and $\rho \to -\infty$ (goods are perfect complements), respectively:
\begin{equation*}
\begin{aligned}
    &\text{Linear:}\\ 
    &\util[\buyer](\allocation[\buyer]) = \sum_{\good \in \goods} \valuation[\buyer][\good] \allocation[\buyer][\good]
\end{aligned}
\quad \vline \quad
\begin{aligned}
    &\text{Cobb-Doulas:}\\
    &\util[\buyer](\allocation[\buyer]) = \prod_{\good \in \goods} \allocation[\buyer][\good]^{\valuation[\buyer][\good]}
\end{aligned}
\quad \vline \quad
\begin{aligned}
    &\text{Leontief:}\\
    & \util[\buyer](\allocation[\buyer]) = \min_{\good : \valuation[\buyer][\good] \neq 0} \frac{\allocation[\buyer][\good]}{\valuation[\buyer][\good]}
\end{aligned}
\end{equation*}

Associated with any consumption $\x \in \choiceset$ are \mydef{prices} $\price \in \R^\numgoods_+$ s.t.\ for all goods $\good \in \goods$, $\price[\good] \geq 0$ denotes the price of good $\good$.
A \mydef{demand correspondence} $\calF: \R^\numgoods_+ \to \choiceset$ takes as input prices $\price \in \R^\numgoods_+$ and outputs a set of consumptions $\calF(\price)$.
If $\calF$ is singleton-valued for all $\price \in \R^\numgoods_+$, then it is called a \mydef{demand function}. 
\sdeni{}{Given a demand function 
$\f$, we define the \mydef{elasticity} $\elastic[f_i][x_j]: \R^\numgoods \to \R$ of output $f_i(\x)$ w.r.t.\@ the $j$th input $x_j$ evaluated at $\x = \y$ as $\elastic[f_i][x_j](\y) = \subdiff[x_j]f_i(\y) \frac{y_j}{f_i(\y)}$.}{}

A good $\good \in \goods$ is said to be a \mydef{substitute (resp.\@ complement) w.r.t. a demand function $\f$} for a good $k \in \goods \setminus \{\good\}$ if the demand $f_\good(\price)$ is increasing (resp.\@ decreasing) in $\price[k]$.
If a buyer's demand $f_\good(\price)$ for good $\good$ is instead weakly increasing (resp.\@ decreasing), good $\good$ is said to be a \mydef{weak substitute} (resp.\@ \mydef{weak complement}) for good $k$.

Next, we define the \mydef{consumer functions} \cite{mas-colell, jehle2001advanced}. 
%
The \mydef{indirect utility function} $\indirectutil[\buyer]: \mathbb{R}^{\numgoods}_+ \times \mathbb{R}_+ \to \mathbb{R}_+$ takes as input prices $\price$ and a budget $\budget[\buyer]$ and outputs the maximum utility the buyer can achieve at that prices within that budget, i.e., $ \indirectutil[\buyer](\price, \budget[\buyer]) = \max_{\allocation[ ] \in \choiceset : \price \cdot \allocation[ ] \leq \budget[\buyer]} \util[\buyer](\allocation[ ])$. 

The \mydef{Marshallian demand} is a correspondence $\marshallians[\buyer]: \R^{\numgoods}_+ \times \R_+ \rightrightarrows \choiceset$ that takes as input prices $\price$ and a budget $\budget[\buyer]$ and outputs the utility-maximizing allocations of goods at that budget, i.e., $ \marshallians[\buyer](\price, \budget[\buyer]) = \argmax_{\allocation[ ] \in \choiceset : \price \cdot \allocation[ ] \leq \budget[\buyer]} \util[\buyer](\allocation[ ])$.

The \mydef{expenditure function} $\expend[\buyer]: \mathbb{R}^{\numgoods}_+ \times \mathbb{R}_+ \to \mathbb{R}_+$ takes as input prices $\price$ and a utility level $\goalutil[\buyer]$ and outputs the minimum amount the buyer must spend to achieve that utility level at those prices, i.e., $\expend[\buyer](\price, \goalutil[\buyer]) = \min_{\allocation[ ] \in \choiceset: \util[\buyer](\allocation[ ]) \geq \goalutil[\buyer]} \price \cdot \allocation[ ]$. 
If the utility function $\util[\buyer]$ is continuous, then the expenditure function is continuous and homogeneous of degree 1 in $\price$ and $\goalutil[\buyer]$ jointly, non-decreasing in $\price$, strictly increasing in $\goalutil[\buyer]$, and concave in $\price$. 

The \mydef{Hicksian demand} is a correspondence $\hicksians[\buyer]: \mathbb{R}^{\numgoods}_+ \times \mathbb{R}_+  \rightrightarrows \mathbb{R}_+$ that takes as input prices $\price$ and a utility level $\goalutil[\buyer]$ and outputs the cost-minimizing allocations of goods at those prices and utility level,  i.e., $\hicksians[\buyer](\price, \goalutil[\buyer]) = \argmin_{\allocation[ ] \in \choiceset: \util[\buyer](\allocation[ ]) \geq \goalutil[\buyer]} \price \cdot \allocation[ ]$.

In this paper, we study buyers whose utilities are s.t.\ Hicksian demand elasticity is well defined, i.e., buyers whose Hicksian demand is singleton-valued, since demands must be unique (and differentiable) for elasticity to be well-defined.%
\footnote{The definition of elasticity can be extended to non-unique and non-differentiable demands with additional care (see, for instance, \citet{cheung2014analyzing}). Nonetheless, we present our analysis in this restricted setting, as it is sufficient to capture all convergent and non-convergent behaviors of \emph{t\^atonnement}, and allows for a simpler analysis.} 
As such, $\marshallians[\buyer](\price, \budget[\buyer]) = \{ \marshallian[\buyer](\price, \budget[\buyer]) \}$ and $\hicksians[\buyer](\price, \budget[\buyer]) = \{ \hicksian[\buyer](\price, \budget[\buyer]) \}$. We note that by \Crefrange{expend-to-budget}{hicksian-marshallian} (\Cref{sec_ap:proofs}), uniqueness of the Hicksian demand implies implies uniqueness of the Marshallian demand and vice-versa.
\amy{if Hicksian demand is singleton-valued, does that imply that Marshallian demand is as well? i ask, b/c the first sentence in this paragraph is about Hicksian demand only, but the second sentence is about both.}\deni{Yes it does.}\amy{do we say this somewhere, or should this be obvious to the reader?} \deni{Added the last sentence to address your comment.}

Fixing a buyer $\buyer \in \buyers$, a good $\good \in \goods$ is said to be a \mydef{gross (resp.\@ net) substitute} for a good $k \in \goods \setminus \{\good\}$ if it is a substitute w.r.t.\ Marshallian (resp.\ Hicksian) demand. 
If $\good$ is instead a weak substitute w.r.t.\ Marshallian (resp.\ Hicksian) demand, then it is called a \mydef{weak gross (resp.\ net) substitute}. Gross and net complements, and their weak counterparts, are defined analogously.

Finally, the \mydef{cross-price elasticity of Marshallian (resp.\@ Hicksian) demand} for good $j$ w.r.t.\@ the price of good $k \neq j$ at price $\price$ and budget $\budget[\buyer]$ (resp.\@ utility level $\goalutil[\buyer]$) is given by $\elastic[{\marshallian[\buyer][\good]}][{\price[k]}](\price, \budget[\buyer])$ (resp.\@ $\elastic[{\hicksian[\buyer][\good]}][{\price[k]}](\price, \goalutil[\buyer])$). 
If $k = j$, then we instead have the \mydef{Hicksian (resp.\ Marshallian) own-price elasticity of demand}. 

\subsection{Fisher Markets}

A \mydef{Fisher market} consists of $\numbuyers$ buyers and $\numgoods$ divisible goods \cite{brainard2000compute}. 
Each buyer $\buyer \in \buyers$ has a budget $\budget[\buyer] \in \mathbb{R}_{+}$ and a utility function $\util[\buyer]: \mathbb{R}_{+}^{\numgoods} \to \mathbb{R}$. 
As is standard in the literature, we assume there is one unit of each good, and one unit of currency available in the market, i.e. $\sum_{\buyer \in \buyers} \budget[\buyer] = 1$  \cite{AGT-book}. 
An instance of a Fisher market is given by a tuple $(\numbuyers, \numgoods, \util, \budget)$, where $\util = (\util[1], \hdots, \util[\numbuyers])$, and $\budget \in \R_{+}^{\numbuyers}$ is the vector of buyer budgets.
We abbreviate as $(\util, \budget)$, when $\numbuyers$ and $\numgoods$ are clear from context.

An \mydef{allocation} $\allocation$ is a map from goods to buyers, represented as a matrix s.t. $\allocation[\buyer][\good] \ge 0$ denotes the amount of good $\good \in \goods$ allocated to buyer $\buyer \in \buyers$. 
Goods are assigned \mydef{prices} $\price \in \mathbb{R}_+^{\numgoods}$. 

\begin{definition}[Competitive Equilibrium]
    A tuple $(\allocation^*, \price^*)$ is said to be a \mydef{competitive (or Walrasian) equilibrium} of a Fisher market $(\util, \budget)$ if 1.~buyers are utility maximizing constrained by their budget, i.e., $\forall \buyer \in \buyers, \allocation[\buyer]^* \in \marshallian[\buyer] (\price^*, \budget[\buyer])$; and 2.~the market clears, i.e., $\forall \good \in \goods,  \price[\good]^* > 0 \Rightarrow \sum_{\buyer \in \buyers} \allocation[\buyer][\good]^* = 1$ and $\price[\good]^* = 0 \Rightarrow\sum_{\buyer \in \buyers} \allocation[\buyer][\good]^* \leq 1$. 
\end{definition}

When the buyers' utility functions in a Fisher market are all of the same type, we qualify the market by the name of the utility function, e.g., a linear Fisher market. A \mydef{mixed CES Fisher market} is a Fisher market which comprises CES buyers with possibly different substitution parameters.
\deni{Maybe remove this:} Considering properties of goods, rather than buyers,
a (Fisher) market satisfies \mydef{gross substitutes} (resp. \mydef{gross complements}) if all pairs of goods in the market are gross substitutes (resp.\@ gross complements).
We define net substitute Fisher markets and net complements Fisher markets similarly.
We refer the reader to \Cref{fig:utility-functions} for a summary of the relationships among various Fisher markets.
A Fisher market is called \mydef{homothetic} if the buyers' utility functions are continuous and homogeneous.



Given a Fisher market $(\util, \budget)$, we define the \mydef{aggregate demand} correspondence $\demand: \R^\numgoods_+ \rightrightarrows \R^\numgoods_+$ at prices $\price$ as the sum of the Marshallian demand at $\price$, given budgets $\budget$, i.e., $\demand(\price) = \sum_{\buyer \in \buyers} \marshallian[\buyer](\price, \budget[\buyer])$.
The \mydef{excess demand} correspondence  $\excess: \mathbb{R}^{\numgoods} \rightrightarrows \mathbb{R}^{\numgoods}$ of a Fisher market $(\util, \budget)$, which takes as input prices and outputs a set of excess demands at those prices, is defined as the difference between the aggregate demand for and the supply of each good: i.e.,
    $\excess(\price) = \demand(\price) - \ones[\numgoods]$ 
where $\ones[\numgoods]$ is the vector of ones of size $\numgoods$, and $\demand(\price) - \bm{1}_{\numgoods} = \{\x - \ones[\numgoods] \mid \forall \x \in \demand(\price)\}$.

The \mydef{discrete \emph{t\^atonnement\/} process} for Fisher markets is a decentralized, natural price adjustment, defined as:
\begin{align*}
    \price[][t+1] = \price[][t] + G(\subgrad^{(t)}) && \text{for } t = 0, 1, 2, \hdots 
    \\
    \subgrad^{(t)} \in \excess(\price[][t])\\
    \price[][0] \in \mathbb{R}^{\numgoods}_+
    \enspace ,
\end{align*}

\noindent
where $G: \mathbb{R}^{\numgoods} \to \mathbb{R}^{\numgoods}$ is a coordinate-wise monotonic function such that, 
for all $\good \in \goods, \x, \y \in \R^m$, if $x_{\good} \geq  y_{\good}$, then $G_{\good}(\bm{x}) \geq G_{\good}(\bm{y})$. 
Intuitively, \emph{t\^atonnement\/} is an auction-like process in which the sellers increase (resp.\@ decrease) the prices of goods whenever the demand (resp.\@ supply) is greater than the supply (resp.\@ demand).

Given a sequence of prices $\{\price^{(\iter)}\}_\iter$, for any $\iter \in \N_+$, we denote buyer $\buyer$'s Marshallian demand for good $\good$, i.e., $\marshallian[\buyer][\good](
\price[][\iter], \budget[\buyer])$, by $\marshallian[\buyer][\good]^{(\iter)}$;
the aggregate demand for good $\good$, i.e., $ \demand[\good](\price[][\iter])$, by $\demand[\good]^{(\iter)}$; and buyer $\buyer$'s Hicksian demand for good $\good$ at utility level 1, i.e.,  $\hicksian[\buyer][\good](\price[][\iter], 1)$, by $\hicksian[\buyer][\good]^{(\iter)}$. 
Additionally, for any fixed $\iter \in \N_+$, which will be clear from context, we define $\pricediff \doteq \price^{(\iter+1)} - \price^{(\iter)}$.

%% file: homothetic.tex
\subsection{Homothetic Fisher Markets}

Suppose that $(\util, \budget)$ is a continuous, concave, and homogeneous Fisher market.
The optimal solutions ($\allocation^*, \price^*)$ to the primal and dual of \mydef{Eisenberg-Gale program}
constitute a competitive equilibrium of $(\util, \budget)$ \cite{devanur2002market,eisenberg1959consensus,jain2005market}:%
\footnote{The dual as presented here was formulated by \citet{goktas2022consumer}.}
\begin{equation*}
\begin{aligned}[c]
    &\text{\textbf{Primal}} \\  
    & \max_{\allocation \in \R_+^{\numbuyers \times \numgoods}} & \sum_{\buyer \in \buyers} \budget[\buyer] \log{\left(\util[\buyer](\allocation[\buyer])\right)} \\
    & \text{subject to} & \sum_{\buyer \in \buyers} \allocation[\buyer][\good]  \leq 1 \quad \quad \forall \good \in \goods\label{feasibility-constraint-EG}
\end{aligned}
\quad \vline \quad
\begin{aligned}[c]
    &\text{\textbf{Dual}} \\  &\min_{\price \in \simplex[\numgoods]} &\sum_{\good \in \goods} \price[\good] + \sum_{\buyer \in \buyers} \left[ \budget[\buyer] \log{ \left( \indirectutil[\buyer](\price, \budget[\buyer]) \right)} - \budget[\buyer] \right]
\end{aligned}
\end{equation*}


Recently, \citet{goktas2022consumer} proposed a convex program that is equivalent to the Eisenberg-Gale convex program, but whose optimal value differs from that of the Eisenberg-Gale convex program by an additive constant.
The authors state their results only for continuous, concave, and homogeneous Fisher markets; however, their proof is valid for \emph{all\/} homothetic Fisher markets, 
including those in which the buyers' utility functions are non-concave:

\begin{theorem}[\citet{goktas2022consumer}]
\label{new-convex}
The optimal solutions $(\allocation^*, \price^*)$ to the following primal and dual convex programs correspond to competitive equilibrium allocations and prices, respectively, of the 
homothetic Fisher market $(\util, \budget)$:
\begin{equation*}
\begin{aligned}[c]
    &\text{\textbf{Primal}} \\  
    & \max_{\allocation \in \R_+^{\numbuyers \times \numgoods}} & \sum_{\buyer \in \buyers} \left[ \budget[\buyer] \log{\util[\buyer]\left(\frac{\allocation[\buyer]}{\budget[\buyer]}\right)} + \budget[\buyer] \right] \\
    & \text{subject to} & \sum_{\buyer \in \buyers} \allocation[\buyer][\good]  \leq 1 \quad \quad \forall \good \in \goods\\
\end{aligned}
\quad \vline \quad
\begin{aligned}[c]
    &\text{\textbf{Dual}} \\  &\min_{\price \in \simplex[\numgoods]} \potential(\price) \doteq \sum_{\good \in \goods} \price[\good] - \sum_{\buyer \in \buyers} \budget[\buyer] \log{ \left( \expend[\buyer] (\price, 1) \right)}
\end{aligned}
\end{equation*}
\end{theorem}


Since the objective function of the primal in \Cref{new-convex} is in general non-concave (i.e., if utilities $\util$ are not concave), strong duality need not hold; however, the dual is still guaranteed to be convex \cite{boyd2004convex}.
This observation suggests that even if the problem of computing competitive equilibrium \emph{allocations\/} is non-concave, the problem of computing competitive equilibrium \emph{prices\/} can still be convex. 
Additionally, since this convex program differs from the Eisenberg-Gale program by an additive constant, we obtain as a corollary that solutions to the Einseberg-Gale program also correspond to competitive equilibria in \emph{all\/} homothetic Fisher markets, including those in which the buyers' utility functions are non-concave.

An interesting property of this convex program is that its dual expresses competitive equilibrium prices via expenditure functions, and just like the Eisenberg-Gale program's dual objective \cite{fisher-tatonnement, devanur2008market}, the gradient of its objective $\potential(\price)$ at any price $\price$ is equal to the negative excess demand in the market at those prices.

\begin{theorem}[\citet{goktas2022consumer}]
\label{excess-demand}
Given any homothetic Fisher market $(\util, \budget)$, the subdifferential of the dual of the program in \Cref{new-convex} at any price $\price$ is equal to the negative excess demand in $(\util, \budget)$ at price $\price$: i.e.,
%
$\subdiff[\price] \potential(\price)
= - \excess(\price)$.
\end{theorem}


\citet{fisher-tatonnement} define a class of markets called \mydef{convex potential function (CPF)} markets.
A market is a CPF market, if there exists a convex potential function $\varphi$ such that $\subdiff[\price] \varphi(\price) = - \excess(\price)$.
A corollary of \Cref{excess-demand} is that all homothetic Fisher markets are CPF markets. 
This in turn implies that mirror descent on $\varphi$ over the unit simplex is equivalent to \emph{t\^atonnement\/} in all homothetic Fisher markets, for some monotone function $G$.
Using this equivalence, we can pick a particular kernel function $h$, and then potentially use \Cref{theorem-devanur} to establish convergence rates for \emph{t\^atonnement}.

Unfortunately, \emph{t\^atonnement\/} does not converge to equilibrium prices in all homothetic Fisher markets, e.g., linear Fisher markets \cite{cole2019balancing}, which suggests the need for additional restrictions on the class of homothetic Fisher markets.
\citet{goktas2022consumer} suggest the maximum absolute value of the  Marshallian price demand elasticity, i.e., $c = \max_{\good, k, \buyer} \max_{(\price, \budget) \in \simplex[\numgoods] \times \simplex[\numbuyers] \times [\numbuyers] } \left\|\elastic[{\marshallian[\buyer][\good]}][{\price[k]}](\price, \budget[\buyer]) \right \|$, as a possible market parameter to use to establish a convergence rate of $O(\nicefrac{(1 + c)}{T})$.
However, 
\citeauthor{cole2008fast}'s [\citeyear{cole2008fast}] results suggest that it is unlikely that Marshallian demand elasticity could be enough, since the proof techniques used in work that makes this assumption require one to quantify the direction of the change in demand as a function of the change in the prices of the other goods, and hence only apply when one assumes WGS or WGC \cite{cole2008fast}.

\section{Market Parameters}

One of the main contributions of this paper is the observation that the maximum absolute value of the price elasticity of Hicksian demand in a homothetic Fisher market is sufficient to analyze the convergence of \emph{t\^atonnement}.
To this end, in this section, we anaylyze Hicksian demand price elasticity, exposit some of its properties in homothetic markets, and argue why it is a natural parameter to consider in the analysis of \emph{t\^atonnement}.

We first note that for Leontief utilities, the Hicksian cross-price elasticity of demand is equal to $0$, while for linear utilities the Hicksian cross-price elasticity of demand is, by convention, $\infty$.%
\footnote{The limit of Hicksian price elasticity of demand as $\rho \to 1$ is not well defined, i.e., if $\rho \to 1^{-}$ the limit is $+\infty$, while if $\rho \to 1^{+}$ the limit is $- \infty$.
However, for linear utilities, as the Hicksian demand for a good can only go up when the price of another good goes up, we set the elasticity of Hicksian price elasticity of demand for linear utilities to be $+\infty$, by convention. We refer the reader to \citet{ramskov2001elasticities} for a primer on elasticity of demand.} 
For Cobb-Douglas utilities, the Hicksian cross-price elasticity of demand is strictly positive and upper bounded by $1$, 
but it is not the same for all pairs of goods.
Note that the behavior of the Hicksian cross-price elasticity of demand is radically different than that of the Marshallian cross-price elasticity of demand, for which the elasticities of linear, Cobb-Douglas, and Leontief utilities are respectively given as $\infty$, $0$, and $-\infty$. 
A taxonomy of utility classes as a function of price elasticity of demand (both Marshallian and Hicksian) is shown in \Cref{fig:taxonomy}.

We start our analysis with following lemma, which shows that the Hicksian price elasticity of demand is constant across all utility levels in homothetic Fisher markets.
This property implies that the Hicksian demand price elasticity at one unit of utility provides sufficient information about the market's reactivity to changes in prices, even without any information about the buyers' utility levels.
This information is crucial when trying to bound the changes in Hicksian demand from one iteration of \emph{t\^atonnement\/} to another, since buyers' utilities can change.%
\footnote{We include all omitted results and proofs in \Cref{sec_ap:proofs}.}

\begin{restatable}{lemma}{lemmahomoelasticity}
\label{lemma:homo_elasticity}
For any Hicksian demand $\hicksian[\buyer]$ associated with a homogeneous utility function $\util[\buyer]$, for all $\good, k \in \goods, \price \in \R^\numgoods_+, \goalutil[\buyer] \in \R_+$, it holds that
    $\elastic[{\hicksian[\buyer][\good]}][{\price[k]}](\price, \goalutil[\buyer]) = \elastic[{\hicksian[\buyer][\good]}][{\price[k]}](\price, 1) = 1$.
\end{restatable}


With the above lemma in hand, we now explain why the Hicksian demand price elasticity%
\footnote{Going forward we refer to the Hicksian price elasticity of demand, as simply Hicksian demand elasticity, because Hicksian price elasticity of demand w.r.t.\ utility level is always 1.}
is a better market parameter by which to analyze the convergence of \emph{t\^atonnement\/} than the Marshallian demand price elasticity. 
\citet{fisher-tatonnement, cheung2014analyzing} use the dual of the Eisenberg-Gale program as a potential to measure the progress that \emph{t\^atonnement\/} makes at each step, for (nested) CES and Leontief utilities.
Under these functional forms, the authors are able to explain a change in the value of the buyers' indirect utilities as a function a change in prices, based on which they bound the change in the second term of the dual $\sum_{\buyer \in \buyers} \budget[\buyer]\log(\indirectutil[\buyer](\price, \budget[\buyer])) - \budget[\buyer]$ from one time period to the next.
Using this bound, they show that \emph{t\^atonnement\/} makes steady progress towards equilibrium.

However, in general homothetic Fisher markets, knowing how much the Marshallian demand for each good changes from one iteration of \emph{t\^atonnement\/} to another does not tell us how much the buyers' utilities change.
More concretely, suppose that the Marshallian demand of a buyer $\buyer$ has changed by an additive vector $\Delta \marshallian[\buyer]$ from time $\iter$ to time $\iter + 1$, then the difference in indirect utilities from one period to another is given by $\util[\buyer](\marshallian[\buyer]^{(\iter+1)}) - \util[\buyer](\marshallian[\buyer]^{(\iter)}) = \util[\buyer](\marshallian[\buyer]^{(\iter)} + \Delta \marshallian[\buyer]) - \util[\buyer](\marshallian[\buyer]^{(\iter)})$.
Without additional information about the utility functions, e.g., Lipschitz continuity, it is impossible to bound this difference, because utilities can change by an unbounded amount from one period to another.
Hence, even if the Marshallian price elasticity of demand and the changes in prices from one period to another were known, it would only allow us to bound the difference in demands, and not the difference in indirect utilities.
To get around this difficulty, one could consider making an assumption about the boundedness of the indirect utility function's price elasticity, or the utility function's Lipschitz-continuity, but such assumptions would not be economically justified, since utility functions are merely representations of preference orderings without any inherent meaning of their own.

We can circumvent this issue by instead looking at the dual of the convex program in \Cref{new-convex}.
In this dual, the indirect utility term is replaced by the expenditure function.
The advantage of this formulation is that if one knows the amount by which prices change from one iteration to the next, as well as the Hicksian demand elasticity, then we can easily bound the change in spending from one period to another. 

The following lemma is crucial to proving the convergence of \emph{t\^atonnement\/}. 
This lemma allows us to bound the changes in buyer spending across all time periods, thereby allowing us to obtain a global convergence rate.
In particular, it shows that the change in spending between two consecutive iterations of \emph{t\^atonnement\/} can be bounded as a function of the prices and the Hicksian demand elasticity.

More formally, suppose that we would like to bound the percentage change in expenditure at one unit of utility from one iteration to another, i.e., $\frac{\expend[\buyer](\price[][\iter+1], 1) - \expend[\buyer](\price[][\iter], 1)}{\expend[\buyer](\price[][\iter], 1)}$, using a first order Taylor expansion of $\expend[\buyer](\price[][\iter] + \pricediff, 1)$ around $\price[][\iter]$.
By Taylor's theorem \cite{graves1927riemann}, we have: $\expend[\buyer](\price[][\iter] + \pricediff, 1) = \expend[\buyer](\price[][\iter], 1) + \left<\grad[\price] \expend[\buyer](\price[][\iter], 1), \pricediff \right> + \nicefrac{1}{2} \left< \grad[\price]^2 \expend[\buyer](\price[][\iter] + c \pricediff, 1) \pricediff, \pricediff \right>$ for some $c \in (0, 1)$.
Re-organizing terms around, we get $\expend[\buyer](\price[][\iter+1], 1) - \expend[\buyer](\price[][\iter], 1) =  \left<\grad[\price] \expend[\buyer](\price[][\iter], 1), \pricediff \right> + \nicefrac{1}{2} \left< \grad[\price]^2 \expend[\buyer](\price[][\iter] + c \pricediff, 1) \pricediff, \pricediff \right>$. 
Dividing both sides by $\expend[\buyer](\price[][\iter], 1)$ we obtain:
\begin{align*}
    \frac{\left<\grad[\price] \expend[\buyer](\price[][\iter], 1), \pricediff \right> + \nicefrac{1}{2} \left< \grad[\price]^2 \expend[\buyer](\price[][\iter] + c \pricediff, 1) \pricediff, \pricediff \right>}{\expend[\buyer](\price[][\iter], 1)} 
\end{align*}
We can now apply Shephard's lemma \cite{shephard}, a corollary of the envelope theorem \cite{afriat1971envelope, milgrom2002envelope}, to the numerator, which allows us to conclude that for all buyers $\buyer \in \buyers$, $\grad[\price] \expend[\buyer](\price, \goalutil[\buyer]) = \hicksian[\buyer](\price, \goalutil[\buyer])$.
Next, using the definition of the expenditure function in the denominator, we obtain the following:
\begin{align}
    &= \frac{\left<\hicksian[\buyer](\price[][\iter], 1), \pricediff \right> }{\left<\hicksian[\buyer](\price[][\iter], 1), \price[][\iter] \right>} + \frac{\nicefrac{1}{2} \left< \grad[\price]^2 \expend[\buyer](\price[][\iter] + c \pricediff, 1) \pricediff, \pricediff \right>}{\expend[\buyer](\price[][\iter], 1)}\label{eq:expend_change}
\end{align}
If the change in prices is bounded, and the Hicksian demand elasticity is known, then one can bound the first term in \Cref{eq:expend_change} with ease.
It remains to be seen if the second term can be bounded.
The following lemma provides an affirmative answer to that question.
In particular, we show that the second-order error term in the Taylor approximation above can be bounded as a function of the maximum absolute value of the Hicksian demand elasticity.
We note that in the following lemma, by \Cref{equiv-def-demand}, the Marshallian demand is unique, because the Hicksian demand is a singleton for bounded elasticity of Hicksian demand. 


\begin{lemma}
\label{lemma:expend_change}
Fix $\buyer \in \buyers$ and $\iter \in \N_+$ and let $\pricediff = \price[][\iter+1] - \price[][\iter]$. Suppose that $\frac{|\pricediff[\good]|}{\price[\good][\iter]} \leq \frac{1}{4}$, then for all buyers $\buyer \in \buyers$, and for some $c \in (0,1)$, it holds that:
\begin{align}
    \left|\frac{\budget[\buyer]}{2} \frac{\left< \grad[\price]^2 \expend[\buyer](\price[][\iter] + c \pricediff, 1) \pricediff, \pricediff \right>}{\expend[\buyer](\price[][\iter], 1)} \right| \leq  \frac{5\elastic }{6}    \sum_{\good} \frac{(\pricediff[\good])^2}{\price[\good][\iter]}  \marshallian[\buyer][\good](\price[][\iter] + c \pricediff, \budget[\buyer])
    \enspace ,
\end{align}
where $\elastic \doteq \max_{\price \in \simplex[\numgoods],  \good, k \in \goods} \left| \elastic[{\hicksian[\buyer][\good]}][{\price[k]}](\price, 1) \right|$.
\end{lemma}

\begin{proof}[Proof of \Cref{lemma:expend_change}]
By Shephard's lemma \cite{shephard} (\Cref{shephard}, \Cref{sec_ap:proofs}), it holds that 
$\left< \grad[\price]^2 \expend[\buyer](\price[][\iter] + c \pricediff, 1) \pricediff, \pricediff \right> = \left< \grad[\price] \hicksian[\buyer](\price[][\iter] + c \pricediff, 1) \pricediff, \pricediff \right>$. Then, it follows that:
\begin{align}
    &\left|\frac{\budget[\buyer]}{2} \frac{\left< \grad[\price]^2 \expend[\buyer](\price[][\iter] + c \pricediff, 1) \pricediff, \pricediff \right>}{\left<\hicksian[\buyer]^t, \price[][\iter] \right>} \right| \notag\\
    &= \left|\frac{\budget[\buyer]}{2} \frac{\left< \grad[\price] \hicksian[\buyer](\price[][\iter] + c \pricediff, 1) \pricediff, \pricediff \right>}{\left<\hicksian[\buyer]^t, \price[][\iter] \right>} \right| && \text{(Shepherd's Lemma)} \notag \\
    &\leq \frac{\budget[\buyer]}{2} \frac{\sum_{\good, k } \left| \pricediff[\good] \right| \left|\subdiff[{\price[k]}] \hicksian[\buyer][\good](\price[][\iter] + c \pricediff, 1) \right| \left|\pricediff[k] \right|}{\left<\hicksian[\buyer]^t, \price[][\iter] \right>} \notag \\
    &= \frac{\budget[\buyer]}{2} \frac{\sum_{\good, k } \left| \pricediff[\good]\right| \sqrt{\left| \subdiff[{\price[k]}] \hicksian[\buyer][\good](\price[][\iter] + c \pricediff, 1) \right|} \sqrt{\left|\subdiff[{\price[k]}] \hicksian[\buyer][\good](\price[][\iter] + c \pricediff, 1) \right|}  \left|\pricediff[k] \right| }{\left<\hicksian[\buyer]^t, \price[][\iter] \right>} \notag\\
    &= \frac{\budget[\buyer]}{2} \frac{\sum_{\good, k } \left|\pricediff[\good]\right| \sqrt{\left|\subdiff[{\price[j]}] \hicksian[\buyer][k](\price[][\iter] + c \pricediff, 1)\right|} \sqrt{\left|\subdiff[{\price[k]}] \hicksian[\buyer][\good](\price[][\iter] + c \pricediff, 1) \right|}  \left|\pricediff[k]\right| }{\left<\hicksian[\buyer]^t, \price[][\iter] \right>} \label{eq:lipschitz_error_bound_homo} 
\end{align}

\noindent 
where the last was obtained from the symmetry of $\grad[\price]^2 \expend[\buyer](\price, \goalutil[\buyer]) = \grad[\price]^2 \expend[\buyer](\price, \goalutil[\buyer])^T$ for all $\buyer \in \buyers, \price \in \R^\numgoods_+, \goalutil[\buyer] \in \R_+$ \cite{mas-colell},  which  combined with Shepherd's lemma gives us $ \grad[\price] \hicksian[\buyer] (\price, \goalutil[\buyer]) =  \grad[\price] \hicksian[\buyer] (\price, \goalutil[\buyer])^T $, i.e., for all $\good, k \in \goods$, $\subdiff[{\price[\good]}] \hicksian[\buyer][k] (\price, \goalutil[\buyer]) = \subdiff[{\price[k]}] \hicksian[\buyer][\good] (\price, \goalutil[\buyer])$. 

Define the Hicksian demand elasticity of buyer $\buyer$ for good $\good$ w.r.t. the price of good $k$ as $\elastic[{\hicksian[\buyer][\good]}][{\price[k]}](\price, \goalutil[\buyer]) = \subdiff[{\price[k]}] \hicksian[\buyer][\good](\price, \goalutil[\buyer]) \frac{\price[k]}{\hicksian[\buyer][\good](\price, \goalutil[\buyer])}$.  Since utility functions are homogeneous, by \Cref{lemma:homo_elasticity} we have for all $\goalutil[\buyer] \in \R_+$, $\elastic[{\hicksian[\buyer][\good]}][{\price[k]}](\price, \goalutil[\buyer]) = \subdiff[{\price[k]}] \hicksian[\buyer][\good](\price, \goalutil[\buyer]) \frac{\price[k]}{\hicksian[\buyer][\good](\price, \goalutil[\buyer])} = \subdiff[{\price[k]}] \hicksian[\buyer][\good](\price, 1) \frac{\price[k]}{\hicksian[\buyer][\good](\price, 1)}$. Re-organizing expressions, we get $ \subdiff[{\price[k]}] \hicksian[\buyer][\good](\price, 1) = \elastic[{\hicksian[\buyer][\good]}][{\price[k]}](\price, 1) \frac{\hicksian[\buyer][\good](\price, 1)}{\price[k]}$. Going back to \Cref{eq:lipschitz_error_bound_homo}, we get:

\begin{align*}
    &= \frac{\budget[\buyer]}{2} \frac{\sum_{\good, k } \left|\pricediff[\good] \right| \sqrt{\left|\elastic[{\hicksian[\buyer][k]}][{\price[\good]}](\price[][\iter] + c \pricediff, 1)  \frac{\hicksian[\buyer][k](\price[][\iter] + c \pricediff, 1) }{\price[\good][\iter] + c \pricediff[\good]} \right|} \sqrt{\left|\elastic[{\hicksian[\buyer][\good]}][{\price[k]}](\price[][\iter] + c \pricediff, 1) \frac{\hicksian[\buyer][\good](\price[][\iter] + c \pricediff, 1)}{\price[k][\iter] + c \pricediff[k]}\right|}  \left|\pricediff[k] \right| }{\left<\hicksian[\buyer]^t, \price[][\iter] \right>} \\ 
    &= \frac{\budget[\buyer]}{2} \frac{\sum_{\good, k } \left|\pricediff[\good] \right| \sqrt{\left|\elastic[{\hicksian[\buyer][k]}][{\price[\good]}](\price[][\iter] + c \pricediff, 1) \right| \frac{\hicksian[\buyer][k](\price[][\iter] + c \pricediff, 1) }{\price[\good][\iter] + c \pricediff[\good]}} \sqrt{\left|\elastic[{\hicksian[\buyer][\good]}][{\price[k]}](\price[][\iter] + c \pricediff, 1) \right| \frac{\hicksian[\buyer][\good](\price[][\iter] + c \pricediff, 1)}{\price[k][\iter] + c \pricediff[k]}}  \left|\pricediff[k] \right| }{\left<\hicksian[\buyer](\price[][\iter], 1), \price[][\iter]\right>}
\end{align*}
Letting $\elastic = \max_{\price \in \R_+^\numgoods, \goalutil[\buyer] \in \R_+, \good, k \in \goods} \left| \elastic[{\hicksian[\buyer][\good]}][{\price[k]}](\price, \goalutil[\buyer]) \right|$. Note that since utility functions are homogeneous, by \Cref{lemma:homo_elasticity} we have $\elastic = \max_{\price \in \R_+^\numgoods, \goalutil[\buyer] \in \R_+, \good, k \in \goods} \left| \elastic[{\hicksian[\buyer][\good]}][{\price[k]}](\price, \goalutil[\buyer]) \right| = \max_{\price \in \R_+^\numgoods,  \good, k \in \goods} \left| \elastic[{\hicksian[\buyer][\good]}][{\price[k]}](\price, 1) \right|$, which gives us:
\begin{align*}
    &\leq \frac{\elastic \budget[\buyer]}{2} \frac{\sum_{\good, k } \left|\pricediff[\good] \right| \sqrt{\frac{\hicksian[\buyer][k](\price[][\iter] + c \pricediff, 1) }{\price[\good][\iter] + c \pricediff[\good]}} \sqrt{ \frac{\hicksian[\buyer][\good](\price[][\iter] + c \pricediff, 1)}{\price[k][\iter] + c \pricediff[k]}}  \left|\pricediff[k] \right| }{\left<\hicksian[\buyer](\price[][\iter], 1), \price[][\iter]\right>}
\end{align*}
Since for all $\good \in \goods$, $\frac{|\pricediff[\good]|}{\price[\good]} \leq \frac{1}{4}$, we have that for all $\good \in \goods$ and for all $c \in [0, 1]$,  $\price[\good][\iter] + c \pricediff[\good]\geq \frac{3}{4} \price[\good][\iter]$, which gives:
\begin{align*}
    &\leq \frac{\elastic \budget[\buyer]}{2}  \frac{\sum_{\good, k } \left|\pricediff[\good] \right| \sqrt{\frac{\hicksian[\buyer][k](\price[][\iter] + c \pricediff, 1) }{\frac{3}{4} \price[\good][\iter]}} \sqrt{ \frac{\hicksian[\buyer][\good](\price[][\iter] + c \pricediff, 1)}{\frac{3}{4} \price[k][\iter]}}  \left|\pricediff[k] \right| }{\left<\hicksian[\buyer](\price[][\iter], 1), \price[][\iter]\right>} \\
    &= \frac{2\elastic \budget[\buyer] }{3}   \frac{\sum_{\good, k } \left|\pricediff[\good] \right| \sqrt{\frac{\hicksian[\buyer][\good](\price[][\iter] + c \pricediff, 1) }{ \price[\good][\iter]}} \sqrt{ \frac{\hicksian[\buyer][k](\price[][\iter] + c \pricediff, 1)}{ \price[k][\iter]}}  \left|\pricediff[k] \right| }{\left<\hicksian[\buyer](\price[][\iter], 1), \price[][\iter]\right>} \\
    &= \frac{2\elastic \budget[\buyer] }{3}   \frac{\sum_{\good, k \in \goods} \sqrt{\frac{\left|\pricediff[\good] \right|^2}{ \price[\good][\iter]} \hicksian[\buyer][\good](\price[][\iter] + c \pricediff, 1) \hicksian[\buyer][k](\price[][\iter] + c \pricediff, 1)  \frac{\left|\pricediff[k] \right|^2}{\price[k][\iter]}}}{\left<\hicksian[\buyer](\price[][\iter], 1), \price[][\iter]\right>} \\
\end{align*}

Applying the AM-GM inequality, i.e., for all $x, y \in \R_+$, $\frac{x+y}{2} \geq \sqrt{xy}$, to the sum inside the numerator above, we obtain:
\begin{align*}
    &\leq \frac{2\elastic \budget[\buyer] }{3}   \frac{\sum_{\good, k \in \goods} \nicefrac{1}{2} \left(\frac{\left|\pricediff[\good] \right|^2}{ \price[\good][\iter]} \hicksian[\buyer][\good](\price[][\iter] + c \pricediff, 1) + \hicksian[\buyer][k](\price[][\iter] + c \pricediff, 1)  \frac{\left|\pricediff[k] \right|^2}{\price[k][\iter]} \right)}{\left<\hicksian[\buyer](\price[][\iter], 1), \price[][\iter]\right>} \\
    &\leq \frac{2\elastic \budget[\buyer] }{3} \frac{\sum_{\good} \frac{(\pricediff[\good])^2}{\price[\good][\iter]}  \hicksian[\buyer][\good](\price[][\iter] + c \pricediff, 1)}{\left<\hicksian[\buyer](\price[][\iter], 1), \price[][\iter]\right>} 
\end{align*}

\noindent
Since for all $\good \in \goods$, $\frac{|\pricediff[\good]|}{\price[\good][\iter]} \leq \frac{1}{4}$, we have for all $c \in [0, 1]$ that $\frac{4}{5} \sum_{\good} \hicksian[\buyer][\good](\price[][\iter], 1)(\price[\good][\iter] + c \pricediff[\good]) \leq \sum_{\good} \hicksian[\buyer][\good](\price[][\iter], 1)\price[\good][\iter]$:
\begin{align*}
    &\leq \frac{2\elastic \budget[\buyer] }{3} \frac{\sum_{\good} \frac{(\pricediff[\good])^2}{\price[\good][\iter]}  \hicksian[\buyer][\good](\price[][\iter] + c \pricediff, 1)}{\frac{4}{5} \sum_{\good} \hicksian[\buyer][\good](\price[][\iter], 1)(\price[\good][\iter] + c \pricediff[\good])} \\
    &= \frac{5\elastic \budget[\buyer] }{6}   \frac{\sum_{\good} \frac{(\pricediff[\good])^2}{\price[\good][\iter]}  \hicksian[\buyer][\good](\price[][\iter] + c \pricediff, 1)}{ \sum_{\good} \hicksian[\buyer][\good](\price[][\iter], 1)(\price[\good][\iter] + c \pricediff[\good])}\\
    &\leq \frac{5\elastic \budget[\buyer] }{6}  \frac{\sum_{\good} \frac{(\pricediff[\good])^2}{\price[\good][\iter]}  \hicksian[\buyer][\good](\price[][\iter] + c \pricediff, 1)}{ \sum_{\good} \hicksian[\buyer][\good](\price[][\iter] + c\pricediff, 1)(\price[\good][\iter] + c \pricediff[\good])} && \text{(\Cref{law-of-demand-corollary}, \Cref{sec_ap:proofs})}\\
    &= \frac{5\elastic }{6}   \sum_{\good} \frac{(\pricediff[\good])^2}{\price[\good][\iter]}  \marshallian[\buyer][\good](\price[][\iter] + c \pricediff, \budget[\buyer]) && \text{(\Cref{equiv-def-demand}, \Cref{sec_ap:proofs})}
\end{align*}
\end{proof}

Because we can bound the change in the expenditure function from one iteration of \emph{t\^atonnement\/} to the next, the Hicksian price elasticity of demand is a better tool with which to analyze the convergence of \emph{t\^atonnement\/} than Marshallian price elasticity of demand.
Additionally, as shown previously by \citet{fisher-tatonnement} (\Cref{kl-divergence-1}, \Cref{sec_ap:proofs}), we can further upper bound the price terms in \Cref{lemma:expend_change} by the KL divergence between the two prices 
In light of \Cref{theorem-devanur}, this result suggests that running mirror descent with KL divergence as the Bregman divergence on the dual of the convex program in \Cref{new-convex} could result in a \emph{t\^atonnement\/} update rule that converges to a competitive equilibrium.  

%% file: tatonnement.tex
\section{Convergence Bounds for Entropic T\^atonnement}
\label{sec:convergence}

In this section, we analyze the rate of convergence of \mydef{entropic t\^atonnement}, which corresponds to the \emph{t\^atonnement\/} process given by mirror descent with weighted entropy as the kernel function, i.e., entropic descent. 
This particular update rule reduces to \Crefrange{exp:subgradient-descent}{exp:subgradient-descent-init}, and has been the focus of previous work \cite{fisher-tatonnement}.
We provide a sketch of the proof used to obtain our convergence rate in this section.
The omitted lemmas and proofs can be found in Appendix \ref{sec_ap:proofs}.

At a high level, our proof follows \citeauthor{fisher-tatonnement}'s [\citeyear{fisher-tatonnement}] proof technique for Leontief Fisher markets \cite{fisher-tatonnement}, although we encounter different lower-level technical challenges in generalizing to homothetic markets.
This proof technique works as follows.
First, we prove that under certain assumptions, the condition required by \Cref{theorem-devanur} holds when $f$ is the convex potential function for homothetic Fisher markets, i.e., $f = \potential$.
For these assumptions to be valid, we need to set $\gamma$ to be greater than a quadratic function of the maximum absolute value of the price elasticity of the Hicksian demand and the maximum Marshallian demand, for all goods throughout the \emph{t\^atonnment\/} process.
Further, since $\gamma$ needs to be set at the outset, we need to upper bound $\gamma$.
To do so, we derive a bound on the maximum demand for any good during \emph{t\^atonnement\/} in all homothetic Fisher markets, which in turn allows us to derive an upper bound on $\gamma$.
Finally, we use \Cref{theorem-devanur} to obtain the convergence rate of $O(\nicefrac{1+ \elastic^2}{t})$. 

The following lemma derives the conditions under which the antecedent of \Cref{theorem-devanur} holds for entropic \emph{t\^atonnement}.

\begin{restatable}{lemma}{ineqdevanur}
\label{ineq-devanur}
Consider a homothetic Fisher market $(\util, \budget)$ and let $\elastic = \max_{\price \in \simplex[\numgoods],  \good, k \in \goods} \left| \elastic[{\hicksian[\buyer][\good]}][{\price[k]}](\price, 1) \right|$. Then, the following holds for entropic \emph{t\^atonnement\/} when run on $(\util, \budget)$:
for all $t \in \N$,
$$    \potential(\price[][t+1]) \leq \lapprox[\potential][{\price[][t+1]}][{\price[][t]}]  + \gamma \divergence[\mathrm{KL}][{\price[][t+1]}][{\price[][t]}] \enspace ,$$
where $\gamma = \left(1 + \max\limits_{\good \in \goods} \left\{\sum\limits_{\buyer \in \buyers} \max\limits_{t \in \N_+} \indirectutil[\buyer](\price[][\iter], \budget[\buyer])  \max\limits_{\price \in \simplex[\numgoods]} \hicksian[\buyer][\good](\price, 1) \right\} \right) \left(6 + \frac{85 \elastic}{12} + \frac{25\elastic^2 }{72} \right)$
\end{restatable}

\begin{proof}[Proof of \Cref{ineq-devanur}]

Fix $\iter \in \N$. First, notice that $\gamma \geq 6$ under the choice of $\gamma$ given as in the Lemma statement. Hence, by \Cref{price-change}, we have $\forall \good \in \goods, \frac{|\pricediff[\good]|}{\price[\good][\iter]} \leq \frac{1}{4}$. With this in mind, we have:

\begin{align*}
    &\potential(\price[][t+1]) - \lapprox[\potential][{\price[][t+1]}][{\price[][t]}]\\
    &= \potential(\price[][t+1]) - \potential(\price[][t]) + \excess(\price[][t]) \cdot \left(\price[][t+1] - \price[][t]\right)\\
    &= \sum_{\good \in \goods} \!\!\!\!\left( \price[\good]^{(\iter)} + \pricediff[\good] \right) - \!\!\! \sum_{\buyer \in \buyers} \budget[\buyer] \log \left( \expend[\buyer](\price[][t+1], 1)\right) - \!\!\!\sum_{\good \in \goods} \!\! \price[\good]^{(\iter)} + \!\!\!\! \sum_{\buyer \in \buyers} \budget[\buyer] \log \left( \expend[\buyer](\price[][t] , 1)\right) + \!\!\!\! \sum_{\good \in \goods} \!\! \excess[\good](\price[][t]) \pricediff[\good]\\
    &= \sum_{\good \in \goods} \left( \price[\good]^{(\iter)} + \pricediff[\good] \right) - \!\!\!\! \sum_{\buyer \in \buyers} \budget[\buyer] \log \left( \expend[\buyer](\price[][t+1] , 1)\right) - \!\!\!\! \sum_{\good \in \goods} \!\! \price[\good]^{(\iter)} \!\!+ \!\!\! \sum_{\buyer \in \buyers} \budget[\buyer] \log \left( \expend[\buyer](\price[][t] , 1)\right) + \!\!\! \sum_{\good \in \goods} (\demand[\good]^{(\iter)} \!\! - 1) \pricediff[\good]\\
    &=  \sum_{\good \in \goods} \pricediff[\good] \demand[\good]^{(\iter)} - \sum_{\buyer \in \buyers} \budget[\buyer] \log \left( \expend[\buyer](\price[][t+1], 1)\right) + \sum_{\buyer \in \buyers} \budget[\buyer] \log \left( \expend[\buyer](\price[][t] , 1)\right) \\
    &= \left<\pricediff, \demand^{(\iter)} \right> + \sum_{\buyer \in \buyers} \budget[\buyer] \log \left( \frac{\expend[\buyer](\price[][t], 1)}{\expend[\buyer](\price[][t+1], 1)}\right)\\
    &=\left<\pricediff, \demand^{(\iter)} \right> + \sum_{\buyer \in \buyers} \budget[\buyer] \log \left( \frac{\expend[\buyer](\price[][t], 1)}{\expend[\buyer](\price[][t], 1) + \expend[\buyer](\price[][t+1], 1) - \expend[\buyer](\price[][t], 1)}\right)\\
    &=\left<\pricediff, \demand^{(\iter)} \right> + \sum_{\buyer \in \buyers} \budget[\buyer] \log \left( 1 - \frac{\expend[\buyer](\price[][t+1], 1) - \expend[\buyer](\price[][t], 1)}{\expend[\buyer](\price[][t], 1)} \left(1+ \frac{\expend[\buyer](\price[][t+1], 1) - \expend[\buyer](\price[][t], 1)}{\expend[\buyer](\price[][t], 1)} \right)^{-1} \right) 
\end{align*}
where the last line is obtained by simply noting that $\forall a, b \in \R, \frac{a}{a+b} = 1 - \frac{b}{a}(1+ \frac{b}{a})^{-1}$. Using \Cref{fix-buyer-bound}, we then obtain:
\begin{align*}
    &\potential(\price[][t+1]) - \lapprox[\potential][{\price[][t+1]}][{\price[][t]}]\\ 
    &\leq \left<\pricediff, \demand^{(\iter)} \!\right> \!+ \!\!\! \sum_{\buyer \in \buyers} \left[ \left(\frac{4}{3}+\frac{20\elastic}{27}\right) \! \! \sum_{l \in \goods} \! \frac{\marshallian[\buyer][l]^{(\iter)}}{\price[l]^{(\iter)}}(\pricediff[l])^2 \! + \! \left(\frac{5 \elastic}{6} \!+ \!\! \frac{25\elastic^2 }{324} \right) \!\! \sum_{l \in \goods} \!\! \frac{\marshallian[\buyer][l](\price[][t] \! + \! c \pricediff, \budget[\buyer])}{\price[l]^{(\iter)}}(\pricediff[l])^2 \! - \!\left< \marshallian[\buyer]^{(\iter)}\!\!, \pricediff \right> \right]\\
    &= \left<\pricediff, \demand^{(\iter)} \!\right> \!+\! \left(\frac{4}{3}+\frac{20\elastic}{27}\right) \!\! \sum_{l \in \goods}\! \frac{\demand[l]^{(\iter)}}{\price[l]^{(\iter)}}(\pricediff[l])^2  + \left(\frac{5 \elastic}{6} + \frac{25\elastic^2 }{324} \right)  \sum_{l \in \goods}\frac{\demand[l](\price[][t] + c \pricediff, \budget)}{\price[l]^{(\iter)}}(\pricediff[l])^2 - \left< \demand^{(\iter)}, \pricediff \right>\\
    \!\!&=  \left(\frac{4}{3}+\frac{20\elastic}{27}\right) \sum_{l \in \goods} \frac{\demand[l]^{(\iter)}}{\price[l]^{(\iter)}}(\pricediff[l])^2  + \left(\frac{5 \elastic}{6} + \frac{25\elastic^2 }{324} \right)\sum_{l \in \goods} \frac{\demand[l](\price[][t] + c \pricediff, \budget)}{\price[l]^{(\iter)}}(\pricediff[l])^2 \\
    \!\!&\leq  \left(\!\frac{4}{3}\!+\!\frac{20\elastic}{27}\!\right)\! \sum_{l \in \goods}\! \! \demand[l]^{(\iter)} \!\left(\frac{9}{2}\right)\! \divergence[\mathrm{KL}][{\price[\good]^{(\iter)} \!+\! \pricediff[\good]}][{\price[\good]^{(\iter)}}] \! + \!\left(\frac{5 \elastic}{6} \!+\! \frac{25\elastic^2 }{324} \right) \! \sum_{l \in \goods} \!\demand[l](\price[][t] \! \! + c \pricediff, \budget) \left(\frac{9}{2} \right) \divergence[\mathrm{KL}][{\price[\good]^{(\iter)} \! \! + \pricediff[\good]}][{\price[\good]^{(\iter)}}]
\end{align*}
\noindent where the last line follows from \Cref{kl-divergence-1}. Continuing,
\begin{align}
    \!\!&= \left(6+\frac{10\elastic}{3}\right) \sum_{l \in \goods} \demand[l]^{(\iter)} \divergence[\mathrm{KL}][{\price[l]^{(\iter)} + \pricediff[l]}][{\price[l]^{(\iter)}}]  + \left(\frac{15 \elastic}{4} + \frac{25\elastic^2 }{72} \right) \sum_{l \in \goods}\demand[l](\price[][t] + c \pricediff, \budget)  \divergence[\mathrm{KL}][{\price[l]^{(\iter)} + \pricediff[l]}][{\price[l]^{(\iter)}}] \notag\\
    \!\!&\leq \max_{\substack{\good \in \goods,\\ t \in \N_+}} \! \left\{\demand[\good][t] \!\right\} \! \left(6\!+\!\frac{10\elastic}{3}\right)\!\sum_{l \in \goods} \!\!\divergence[\mathrm{KL}][{\price[l]^{(\iter)} \!\! +\pricediff[l]}][{\price[l]^{(\iter)}}]  \!+\! \notag \\
    &\max_{\substack{\good \in \goods,\\ t \in \N_+,\\ c \in [0,1]}} \left\{\demand[\good](\price[][t] \! \! + c \pricediff)\right\}\! \left(\frac{15 \elastic}{4} \!+\! \frac{25\elastic^2 }{72} \right) \!\! \sum_{l \in \goods}\!\! \divergence[\mathrm{KL}][{\price[l]^{(\iter)} \!\! + \pricediff[l]}][{\price[l]^{(\iter)}}] \notag \\
    \!\!&\leq \max_{\good \in \goods, t \in \N_+, c \in [0,1]} \left\{\demand[\good](\price[][t] + c \pricediff)\right\} \left(6 + \frac{85 \elastic}{12} + \frac{25\elastic^2 }{72} \right)  \divergence[\mathrm{KL}][{\price[][t] + \pricediff}][{\price[][t]}]\label{eq:intermediate_bound_bregman_smooth_fisher}
\end{align}

\sdeni{}{
By \Cref{equiv-def-demand}, we can re-write the aggregate demand for all goods $\good \in \goods$, as follows: 
\begin{align*}
    \demand[\good](\price[][t] + c \pricediff) &= \sum_{\buyer \in \buyers} \marshallian[\buyer][\good]((1-c) \price[][t] + c \price[][\iter + 1], \budget[\buyer])\\
    &= \sum_{\buyer \in \buyers} \frac{\hicksian[\buyer][\good]((1-c) \price[][t] + c \price[][\iter + 1], 1) \budget[\buyer]}{\expend[\buyer]((1-c) \price[][t] + c \price[][\iter + 1], 1)}
\end{align*}

Now, by Danskin's maximum theorem \cite{danskin1966thm}, we know that the expenditure function is concave in prices, that is, for all $c \in [0, 1]$, we have $(1-c) \expend[\buyer](\price[][t], 1) + c \expend[\buyer](\price[][\iter + 1], 1) \leq \expend[\buyer]((1-c) \price[][t] + c \price[][\iter + 1], 1)$. Hence, continuing we have for all $\good \in \goods$:
\begin{align*}
    \demand[\good](\price[][t] + c \pricediff) &= \sum_{\buyer \in \buyers} \frac{\hicksian[\buyer][\good]((1-c) \price[][t] + c \price[][\iter + 1], 1) \budget[\buyer]}{\expend[\buyer]((1-c) \price[][t] + c \price[][\iter + 1], 1)}\\
    &\leq \sum_{\buyer \in \buyers} \frac{\hicksian[\buyer][\good]((1-c) \price[][t] + c \price[][\iter + 1], 1) \budget[\buyer]}{(1-c) \expend[\buyer](\price[][t], 1) + c \expend[\buyer](\price[][\iter + 1], 1)}
\end{align*}

Further, by Lemma 5 of \citet{goktas2022consumer}, since the expenditure function is homogeneous of degree $0$ in prices, notice that we have for all $\good \in \goods$, $\max_{\price \in \R^\numgoods_+ / \{ \zeros \}} \hicksian[\buyer][\good](\price, 1) = \max_{\price \in \simplex[\numgoods]} \hicksian[\buyer][\good](\price, 1)$. Note that $\max_{\price \in \simplex[\numgoods]} \hicksian[\buyer][\good](\price, 1)$ is well-defined since $\simplex[\numgoods]$ is compact, $\hicksian[\buyer][\good](\price, 1)$ exists for all $\price \in \R^\numgoods_+$, and is by Berge's maximum theorem \cite{berge1997topological} continuous in homothetic Fisher markets. Since by the entropic t\^atonnement update rule for all time-steps $\iter \in \N_+$, and goods $\good \in \goods$, $\price[\good][\iter] > 0$, we then have $\hicksian[\buyer][\good]((1-c) \price[][t] + c \price[][\iter + 1], 1) \leq \max_{\price \in \simplex[\numgoods]} \hicksian[\buyer][\good](\price, 1)$. Hence, continuing, we have for all $\good \in \goods$:
\begin{align*}
    \demand[\good](\price[][t] + c \pricediff) &\leq \sum_{\buyer \in \buyers} \frac{\max_{\price \in \simplex[\numgoods]} \hicksian[\buyer][\good](\price, 1) \budget[\buyer]}{(1-c) \expend[\buyer](\price[][t], 1) + c \expend[\buyer](\price[][\iter + 1], 1)}
\end{align*}
Taking a maximum over $c \in [0, 1]$ and $\iter \in \N_+$, and $\good \in \goods$, we have for all goods $\good \in \goods$:
\begin{align*}
    \max_{\good \in \goods, t \in \N_+, c \in [0,1]} \demand[\good](\price[][t] + c \pricediff) &\leq \max_{\good \in \goods, t \in \N_+, c \in [0,1]}  \sum_{\buyer \in \buyers} \frac{\max_{\price \in \simplex[\numgoods]} \hicksian[\buyer][\good](\price, 1) \budget[\buyer]}{(1-c) \expend[\buyer](\price[][t], 1) + c \expend[\buyer](\price[][\iter + 1], 1)}\\
    &\leq \max_{\good \in \goods}  \sum_{\buyer \in \buyers} \frac{\max_{\price \in \simplex[\numgoods]} \hicksian[\buyer][\good](\price, 1) \budget[\buyer]}{\min_{t \in \N_+, c \in [0, 1]} \{(1-c) \expend[\buyer](\price[][t], 1) + c \expend[\buyer](\price[][\iter + 1], 1) \}}\\
    &= \max_{\good \in \goods}  \sum_{\buyer \in \buyers} \frac{\max_{\price \in \simplex[\numgoods]} \hicksian[\buyer][\good](\price, 1) \budget[\buyer]}{\min_{t \in \N_+}\{\min\{\expend[\buyer](\price[][t], 1), \expend[\buyer](\price[][\iter + 1], 1) \}\}}\\
    &= \max_{\good \in \goods}  \sum_{\buyer \in \buyers} \frac{\max_{\price \in \simplex[\numgoods]} \hicksian[\buyer][\good](\price, 1) \budget[\buyer]}{\min_{t \in \N_+} \expend[\buyer](\price[][\iter], 1)}\\
    &= \max_{\good \in \goods}  \sum_{\buyer \in \buyers} \max_{t \in \N_+} \frac{\max_{\price \in \simplex[\numgoods]} \hicksian[\buyer][\good](\price, 1) \budget[\buyer]}{ \expend[\buyer](\price[][\iter], 1)}\\
    &= \max_{\good \in \goods} \sum_{\buyer \in \buyers}  \max_{t \in \N_+} \indirectutil[\buyer](\price[][\iter], \budget[\buyer]) \max_{\price \in \simplex[\numgoods]} \hicksian[\buyer][\good](\price, 1) 
\end{align*}
where the last line follows from Corollary 1, Appendix A of \citet{goktas2022consumer}.
Plugging the above bound into \Cref{eq:intermediate_bound_bregman_smooth_fisher}, we then obtain the following bound which implies the result:
\begin{align*}
    &\potential(\price[][t+1]) - \lapprox[\potential][{\price[][t+1]}][{\price[][t]}] \\
    &\leq \max_{\good \in \goods} \left\{\sum_{\buyer \in \buyers} \max_{t \in \N_+} \indirectutil[\buyer](\price[][\iter], \budget[\buyer])  \max_{\price \in \simplex[\numgoods]} \hicksian[\buyer][\good](\price, 1) \right\} \left(6 + \frac{85 \elastic}{12} + \frac{25\elastic^2 }{72} \right)  \divergence[\mathrm{KL}][{\price[][t] + \pricediff}][{\price[][t]}] 
\end{align*}
}

\end{proof}

\sdeni{}{
For the above lemma to be applied in conjuction with \Cref{theorem-devanur}, we have to ensure that the quantity $\max_{\good \in \goods, t \in \N_+} \left\{\sum_{\buyer \in \buyers} \indirectutil[\buyer](\price[][\iter], \budget[\buyer])  \max_{\price \in \simplex[\numgoods]} \hicksian[\buyer][\good](\price, 1) \right\}$ is bounded throughout entropic \emph{t\^atonnement} for homothetic Fisher markets. To understand the relevance of this bound, we note that this quantity is an upper bound to the aggregate demand, that is:

\begin{align*}
     \demand[\good][\iter] &=\sum_{\buyer \in \buyers} \marshallian[\buyer][\good]^{(\iter)}\\ 
     &= \sum_{\buyer \in \buyers} \hicksian[\buyer][\good](\price[][\iter], \indirectutil[\buyer](\price[][\iter], \budget[\buyer]))\\
     &= \sum_{\buyer \in \buyers} \indirectutil[\buyer](\price[][\iter], \budget[\buyer]) \hicksian[\buyer][\good](\price[][\iter], 1) && \text{(Lemma 5 of \citet{goktas2022consumer})}\\ 
     &\leq \sum_{\buyer \in \buyers} \min_{\iter \in \N_+} \indirectutil[\buyer](\price[][\iter], \budget[\buyer]) \max_{\price \in \simplex[\numgoods]} \hicksian[\buyer][\good](\price, 1)
\end{align*}

As such, proving an upper bound to it implies the excess demand is bounded throughout entropic \emph{t\^atonnement}, which in turn implies Lipschitz-smoothness (and hence Bregman-smoothness for any choice of strongly convex kernel function) of the dual of our convex program over all trajectories of entropic \emph{t\^atonnement}. The following lemma establishes such a bound and shows that it depends on the initial choice of price $\price[][0]$, and the maximum possible Hicksian demand to obtain one unit of utility.
}

\sdeni{}{
\begin{restatable}
[Bounded Indirect Utility for Homothetic Fisher Markets]{lemma}{lemmachomothetic}\label{lemma:c_homothetic}
If entropic \emph{t\^atonnement\/} is run on a homothetic Fisher market $(\util, \budget)$, then, for all $\iter \in \N_+$,  the following bound holds:
\begin{align*}
    \indirectutil[\buyer](\price[][\iter], \budget[\buyer]) \max_{\price \in \simplex[\numgoods]} \hicksian[\buyer][\good](\price, 1) \leq \indirectutil[\buyer](\price[][0], \budget[\buyer]) \max_{\price \in \simplex[\numgoods]} \hicksian[\buyer][\good](\price, 1) + 2 \max\limits_{\substack{\price, \otherprice \in \simplex[\numgoods]\\ k \in \goods : \hicksian[\buyer][k](\price, 1) > 0}} \left\{ \frac{\hicksian[\buyer][\good](\otherprice, 1)^2}{\hicksian[\buyer][k](\price, 1)^2} \right\}
\end{align*}
\end{restatable}
\begin{proof}[Proof of \Cref{lemma:c_homothetic}]
Fix a buyer $\buyer \in \buyers$ and good $\good \in \goods$.
First, note that since by Lemma 5 of \citet{goktas2022consumer}, since the expenditure function is homogeneous of degree $0$ in prices, we have for all $\good \in \goods$, $\max_{\price \in \R^\numgoods_+ / \{ \zeros \}} \hicksian[\buyer][\good](\price, 1) = \max_{\price \in \simplex[\numgoods]} \hicksian[\buyer][\good](\price, 1)$. In addition, note that $\max_{\price \in \simplex[\numgoods]} \hicksian[\buyer][\good](\price, 1)$ is well-defined since $\simplex[\numgoods]$ is compact, $\hicksian[\buyer][\good](\price, 1)$ exists for all $\price \in \R^\numgoods_+$, and is by Berge's maximum theorem \cite{berge1997topological} continuous in homothetic Fisher markets. Further, by the entropic t\^atonnement update rule for all time-steps $\iter \in \N_+$, and goods $\good \in \goods$, $\price[\good][\iter] > 0$, we then have $\hicksian[\buyer][\good](\price[][t], 1) \leq \max_{\price \in \simplex[\numgoods]} \hicksian[\buyer][\good](\price, 1)$.

We now proceed to prove the claim of the lemma by induction on $\iter$.

\paragraph{Base case: $\iter = 0$.}
Since $\max\limits_{\substack{\price, \otherprice \in \simplex[\numgoods]\\ k \in \goods : \hicksian[\buyer][k](\price, 1) > 0}} \left\{ \frac{\hicksian[\buyer][\good](\otherprice, 1)^2}{\hicksian[\buyer][k](\price, 1)^2} \right\} \geq 0$, by definition, we have $$ \indirectutil[\buyer](\price[][0], \budget[\buyer]) \max_{\price \in \simplex[\numgoods]} \hicksian[\buyer][\good](\price, 1) \leq \indirectutil[\buyer](\price[][0], \budget[\buyer]) \max_{\price \in \simplex[\numgoods]} \hicksian[\buyer][\good](\price, 1) +  2 \max\limits_{\substack{\price, \otherprice \in \simplex[\numgoods]\\ k \in \goods : \hicksian[\buyer][k](\price, 1) > 0}} \left\{ \frac{\hicksian[\buyer][\good](\otherprice, 1)^2}{\hicksian[\buyer][k](\price, 1)^2} \right\} \enspace .$$

\paragraph{Inductive hypothesis.}

Suppose that for any $\iter \in \N$, we have: 
$$\indirectutil[\buyer](\price[][\iter], \budget[\buyer]) \max_{\price \in \simplex[\numgoods]} \hicksian[\buyer][\good](\price, 1) \leq \indirectutil[\buyer](\price[][0], \budget[\buyer]) \max_{\price \in \simplex[\numgoods]} \hicksian[\buyer][\good](\price, 1) + 2 \max\limits_{\substack{\price, \otherprice \in \simplex[\numgoods]\\ k \in \goods : \hicksian[\buyer][k](\price, 1) > 0}} \left\{ \frac{\hicksian[\buyer][\good](\otherprice, 1)^2}{\hicksian[\buyer][k](\price, 1)^2} \right\}$$

\paragraph{Inductive step.}

We will show that the inductive hypothesis holds for $\iter + 1$.
We proceed with a proof by cases.

\paragraph{Case 1:} $\marshallian[\buyer][\good]^{(\iter)} \geq   \max\limits_{\substack{\price, \otherprice \in \simplex[\numgoods]\\ k \in \goods : \hicksian[\buyer][k](\price, 1) > 0}} \left\{ \frac{\hicksian[\buyer][\good](\otherprice, 1)}{\hicksian[\buyer][k](\price, 1)} \right\}$. 

For all $k \in \goods$, we have:
\begin{align*}
    \marshallian[\buyer][k]^{(\iter)} &= \frac{\hicksian[\buyer][k]^{(\iter)}}{\hicksian[\buyer][\good][\iter]} \marshallian[\buyer][\good]^{(\iter)} && \text{(\Cref{equiv-def-demand}, \Cref{sec_ap:proofs})}\\
    &\geq \frac{\hicksian[\buyer][k]^{(\iter)}}{\hicksian[\buyer][\good][\iter]} \max_{\substack{\price, \otherprice \in \simplex[\numgoods]\\ k \in \goods : \hicksian[\buyer][k](\price, 1) > 0}} \left\{ \frac{\hicksian[\buyer][\good](\otherprice, 1)}{\hicksian[\buyer][k](\price, 1)} \right\}\\
    &\geq \frac{\hicksian[\buyer][k]^{(\iter)}}{\hicksian[\buyer][\good][\iter]} \frac{\hicksian[\buyer][\good][\iter]}{\hicksian[\buyer][k]^{(\iter)}} = 1
\end{align*}
where the penultimate line follows from the case hypothesis.

The above means that the price of all goods will increase in the next time period, i.e., $\forall k \in \goods, \price[k][\iter+1] \geq \price[k][\iter]$ which implies that $\expend[\buyer](\price[][\iter+1], 1) \geq \expend[\buyer](\price[][\iter], 1) \geq 0$. In addition, note that the expenditure is positive since prices reach 0 only asymptotically under entropic \emph{t\^atonnement}. Which gives us:
\begin{align*}
    \frac{\budget[\buyer]}{\expend[\buyer](\price[][\iter+1], 1)} &\leq \frac{\budget[\buyer]}{\expend[\buyer](\price[][\iter], 1)}\\
    \indirectutil[\buyer](\price[][\iter + 1 ], \budget[\buyer]) &\leq  \indirectutil[\buyer](\price[][\iter], \budget[\buyer]) && \text{(Corollary 1 of \citet{goktas2022consumer})}
\end{align*}

Multiplying both sides by $\max_{\price \in \simplex[\numgoods]} \hicksian[\buyer][\good](\price, 1)$, we have for all $\good \in \goods$:
\begin{align*}
     \indirectutil[\buyer](\price[][\iter + 1 ], \budget[\buyer]) \max_{\price \in \simplex[\numgoods]} \hicksian[\buyer][\good](\price, 1) &\leq  \indirectutil[\buyer](\price[][\iter], \budget[\buyer]) \max_{\price \in \simplex[\numgoods]} \hicksian[\buyer][\good](\price, 1)\\
    &= \indirectutil[\buyer](\price[][0], \budget[\buyer]) \max_{\price \in \simplex[\numgoods]} \hicksian[\buyer][\good](\price, 1) + 2 \max_{\substack{\price, \otherprice \in \simplex[\numgoods]\\ k \in \goods : \hicksian[\buyer][k](\price, 1) > 0}} \left\{ \frac{\hicksian[\buyer][\good](\otherprice, 1)^2}{\hicksian[\buyer][k](\price, 1)^2} \right\}
\end{align*}

where the last line follows by the induction hypothesis. 

\paragraph{Case 2:} $\marshallian[\buyer][\good]^{(\iter)} <   \max_{\substack{\price, \otherprice \in \simplex[\numgoods]\\ k \in \goods : \hicksian[\buyer][k](\price, 1) > 0}} \left\{ \frac{\hicksian[\buyer][\good](\otherprice, 1)}{\hicksian[\buyer][k](\price, 1)} \right\}$. 

For all $k \in \goods$, we have:
\begin{align*}
    \marshallian[\buyer][k]^{(\iter)} &= \frac{\hicksian[\buyer][k]^{(\iter)}}{\hicksian[\buyer][\good][\iter]} \marshallian[\buyer][\good]^{(\iter)}\\
    &\leq \frac{\hicksian[\buyer][k]^{(\iter)}}{\hicksian[\buyer][\good][\iter]} \max\limits_{\substack{\price, \otherprice \in \simplex[\numgoods]\\ k \in \goods : \hicksian[\buyer][k](\price, 1) > 0}} \left\{ \frac{\hicksian[\buyer][\good](\otherprice, 1)}{\hicksian[\buyer][k](\price, 1)} \right\}\\
    &= \frac{\hicksian[\buyer][k]^{(\iter)}}{\hicksian[\buyer][\good][\iter]} \frac{\hicksian[\buyer][\good][\iter]}{\hicksian[\buyer][k]^{(\iter)}} = 1
\end{align*}
where the penultimate line follows from the case hypothesis.

The above means that prices of all goods will decrease in the next time period. Now, note that regardless of the aggregate demand $\demand[][\iter]$ at time $\iter \in \N$, prices can decrease at most by a factor of $e^{-\frac{1}{5}} \geq \nicefrac{1}{2}$, that is, for all $\good \in \goods$
\begin{align*}
    \price[\good][\iter+1] &= \price[\good][\iter]\exp\left\{\frac{\excess[\good](\price[][\iter])}{\gamma}\right\}\\
    &= \price[\good][\iter] \exp\left\{\frac{\demand[\good][\iter] - 1}{\gamma}\right\} \\
    &\geq \price[\good][\iter] \exp\left\{\frac{- 1}{\gamma}\right\} \\
    &\geq \price[\good][\iter] \exp\left\{\frac{-1}{5 \max\limits_{\substack{t \in \N \\ \good \in \goods}} \{1, \demand[\good][\iter]\}}\right\}\\
    &\geq \price[\good][\iter] \exp\left\{\frac{-1}{5}\right\}\\
    &\geq \price[\good][\iter] e^{-\frac{1}{5}} \geq \frac{1}{2} \price[\good][\iter] 
\end{align*}

Now, notice that we have $\expend[\buyer](\price[][\iter+1], 1) \geq \expend[\buyer](\frac{1}{2} \price[][\iter], 1) = \frac{1}{2} \expend[\buyer](\price[][\iter], 1) \geq 0$, since the expenditure of the buyer decreases the most when the prices of all goods decrease simultaneously and the expenditure function is homogeneous of degree $1$ in prices. In addition, note that the expenditure is positive since prices reach 0 only asymptotically under entropic \emph{t\^atonnement}. Hence, we have:

\begin{align*}
    \frac{\budget[\buyer]}{\expend[\buyer](\price[][\iter+1], 1)} &\leq 2\frac{\budget[\buyer]}{\expend[\buyer](\price[][\iter], 1)}\\
    \indirectutil[\buyer](\price[][\iter + 1 ], \budget[\buyer])  &\leq  2\indirectutil[\buyer](\price[][\iter], \budget[\buyer]) && \text{(Corollary 1 of \citet{goktas2022consumer})}
    \end{align*}

Multiplying both sides by $\hicksian[\buyer][\good][\iter]$, and applying \Cref{equiv-def-demand}, we have for all $\good \in \goods$:
    \begin{align*}
    \indirectutil[\buyer](\price[][\iter + 1 ], \budget[\buyer])  \hicksian[\buyer][\good][\iter] &\leq  2\marshallian[\buyer][\good]^{(\iter)}\\
    \indirectutil[\buyer](\price[][\iter + 1 ]) \hicksian[\buyer][\good][\iter] &\leq  2\max_{\substack{\price, \otherprice \in \simplex[\numgoods]\\ k \in \goods : \hicksian[\buyer][k](\price, 1) > 0}} \left\{ \frac{\hicksian[\buyer][\good](\otherprice, 1)}{\hicksian[\buyer][k](\price, 1)} \right\} && \text{(Case hypothesis)}
    \end{align*}

Now, taking a minimum over all $\good \in \goods$ s.t. $\hicksian[\buyer][\good][\iter] > 0$, we have
    \begin{align*}
    \indirectutil[\buyer](\price[][\iter + 1 ]) \min_{k \in \goods: \hicksian[\buyer][k]^{(\iter)} > 0} \hicksian[\buyer][k]^{(\iter)} &\leq  2\max_{\substack{\price, \otherprice \in \simplex[\numgoods]\\ k \in \goods : \hicksian[\buyer][k](\price, 1) > 0}} \left\{ \frac{\hicksian[\buyer][\good](\otherprice, 1)}{\hicksian[\buyer][k](\price, 1)} \right\}\\
    \indirectutil[\buyer](\price[][\iter + 1 ]) \min_{\substack{\price \in \simplex[\numgoods]\\k \in \goods: \hicksian[\buyer][k](\price, 1) > 0}} \hicksian[\buyer][k](\price, 1) &\leq  2\max_{\substack{\price, \otherprice \in \simplex[\numgoods]\\ k \in \goods : \hicksian[\buyer][k](\price, 1) > 0}} \left\{ \frac{\hicksian[\buyer][\good](\otherprice, 1)}{\hicksian[\buyer][k](\price, 1)} \right\}\\
    \indirectutil[\buyer](\price[][\iter + 1 ])  &\leq  \frac{2}{\min\limits_{\substack{\price \in \simplex[\numgoods]\\k \in \goods: \hicksian[\buyer][k](\price, 1) > 0}} \hicksian[\buyer][k](\price, 1)} \max_{\substack{\price, \otherprice \in \simplex[\numgoods]\\ k \in \goods : \hicksian[\buyer][k](\price, 1) > 0}}  \left\{ \frac{\hicksian[\buyer][\good](\otherprice, 1)}{\hicksian[\buyer][k](\price, 1)} \right\} 
\end{align*}

Finally, multiplying both sides by $\max_{\otherprice \in \simplex[\numgoods]} \hicksian[\buyer][\good](\otherprice, 1)$, we have:
\begin{align*}
    \indirectutil[\buyer](\price[][\iter + 1 ]) \max_{\otherprice \in \simplex[\numgoods]} \hicksian[\buyer][\good](\otherprice, 1) &\leq  2\frac{\max_{\otherprice \in \simplex[\numgoods]} \hicksian[\buyer][\good](\otherprice, 1)}{\min\limits_{\substack{\price \in \simplex[\numgoods]\\k \in \goods: \hicksian[\buyer][k](\price, 1) > 0}} \hicksian[\buyer][k](\price, 1)} \max_{\substack{\price, \otherprice \in \simplex[\numgoods]\\ k \in \goods : \hicksian[\buyer][k](\price, 1) > 0}} \left\{ \frac{\hicksian[\buyer][\good](\otherprice, 1)}{\hicksian[\buyer][k](\price, 1)} \right\} \\
    &=  2\left(\max_{\substack{\price, \otherprice \in \simplex[\numgoods]\\ k \in \goods : \hicksian[\buyer][k](\price, 1) > 0}} \left\{ \frac{\hicksian[\buyer][\good](\otherprice, 1)}{\hicksian[\buyer][k](\price, 1)} \right\} \right) \left(\max_{\substack{\price, \otherprice \in \simplex[\numgoods]\\ k \in \goods : \hicksian[\buyer][k](\price, 1) > 0}} \left\{ \frac{\hicksian[\buyer][\good](\otherprice, 1)}{\hicksian[\buyer][k](\price, 1)} \right\} \right) \\
    &=  2\max_{\substack{\price, \otherprice \in \simplex[\numgoods]\\ k \in \goods : \hicksian[\buyer][k](\price, 1) > 0}} \left\{ \frac{\hicksian[\buyer][\good](\otherprice, 1)}{\hicksian[\buyer][k](\price, 1)} \right\}^2 \\
    &\leq \indirectutil[\buyer](\price[][0], \budget[\buyer]) \max_{\price \in \simplex[\numgoods]} \hicksian[\buyer][\good](\price, 1) + 2 \max\limits_{\substack{\price, \otherprice \in \simplex[\numgoods]\\ k \in \goods : \hicksian[\buyer][k](\price, 1) > 0}} \left\{ \frac{\hicksian[\buyer][\good](\otherprice, 1)^2}{\hicksian[\buyer][k](\price, 1)^2} \right\}
\end{align*}

Hence, the inductive hypothesis holds for $\iter +1$. Putting it all together, we have, for all $\iter \in \N$:
\begin{align*}
    \indirectutil[\buyer](\price[][\iter], \budget[\buyer]) \max_{\price \in \simplex[\numgoods]} \hicksian[\buyer][\good](\price, 1) \leq \indirectutil[\buyer](\price[][0], \budget[\buyer]) \max_{\price \in \simplex[\numgoods]} \hicksian[\buyer][\good](\price, 1) + 2  \max\limits_{\substack{\price, \otherprice \in \simplex[\numgoods]\\ k \in \goods : \hicksian[\buyer][k](\price, 1) > 0}} \left\{ \frac{\hicksian[\buyer][\good](\otherprice, 1)^2}{\hicksian[\buyer][k](\price, 1)^2} \right\}
\end{align*}

\end{proof}
}

Combining \Cref{ineq-devanur}, and \Cref{lemma:c_homothetic} with \Cref{theorem-devanur}, we obtain our main result, namely a worst-case convergence rate of $O(\nicefrac{(1 + \elastic^2)}{t})$ for entropic \emph{t\^atonnement\/} in homothetic Fisher markets.

\begin{theorem}
\label{main-convergence-thm}
Suppose $(\util, \budget)$ is a homothetic Fisher market and $\elastic = \max_{\price \in \simplex[\numgoods],  \good, k \in \goods} \left| \elastic[{\hicksian[\buyer][\good]}][{\price[k]}](\price, 1) \right|$. Then, the following holds for entropic \emph{t\^atonnement}:
for all $t \in \N$,
\begin{align}
    \potential(\price[][t]) - \potential(\price^*) \leq \frac{\gamma \divergence[\mathrm{KL}][\price^*][\price^0]}{t}
    \enspace ,
\end{align}
where $\gamma = \left(1 + \max\limits_{\good \in \goods} \sum_{\buyer \in \buyers} \left[
\indirectutil[\buyer](\price[][0], \budget[\buyer]) \max\limits_{\otherprice \in \simplex[\numgoods]} \hicksian[\buyer][\good](\otherprice, 1)  + 2 \max\limits_{\substack{\price, \otherprice \in \simplex[\numgoods]\\ k \in \goods : \hicksian[\buyer][k](\price, 1) > 0}} \left\{ \frac{\hicksian[\buyer][\good](\otherprice, 1)^2}{\hicksian[\buyer][k](\price, 1)^2} \right\} \right] \right) \left(6 + \frac{85 \elastic}{12} + \frac{25\elastic^2 }{72} \right) 
$. 
\end{theorem}

    

\sdeni{Note, however, that for 
\Cref{main-convergence-thm} to hold, the step size $\gamma$ needs to be known in advance, so that it can be set at the start of the entropic \emph{t\^atonnement\/} process.
In other words, we need to upper bound the aggregate demand for all goods all throughout t\^atonnement.
If the number of iterations $T$ for which \emph{t\^atonnement\/} is to be run were known at the outset, then one could use a naive bound of $\max_{t \in \N} \demand[\good]^{(\iter)} \leq e^{\nicefrac{T}{5}} \frac{\sum_{\buyer \in \buyers} \budget[\buyer]}{\min_{\good \in \goods} \price[\good]^0}$.
This bound follows by induction, since for the first iteration the prices of the goods cannot change by more than $\left|\nicefrac{\pricediff[\good]}{\price[\good]^{(\iter)}} \right| \leq \frac{1}{4}$, and then, for all other iterations, the price for a good can decrease at most by a factor of $e^{-\nicefrac{1}{5}}$ (\Cref{price-change}), ensuring that the learning rate is greater than the aggregate demand in the next iteration.
This time-dependent upper bound could then be combined with the doubling trick to extend the convergence result to an unknown time horizon \cite{auer1995gambling}.
In addition, for markets with homothetic \sdeni{}{net} complements, we provide the following bound, which holds even when the \emph{t\^atonnement} time horizon is not known. 
We note that if better upper bounds \amy{on what?} are found, they can be used to improve our convergence result \amy{what does ``our'' convergence result mean? Lemma 4.3? or something else?} by a constant factor.}{}

%% file: conclusion.tex
\section{Conclusion}

We identified the maximum absolute value of the Hicksian price elasticity of demand as a sufficient parameter by which to analyze convergent and non-convergent behavior of \emph{t\^atonnement\/} in homothetic Fisher markets. 
We then showed that together with the KL divergence associated with a change in prices, we can use it to bound the percentage change in the expenditure of one unit of utility, assuming bounded price changes.
This observation motivated us to consider analyzing the convergence of mirror descent with KL divergence on a recently proposed \cite{goktas2022consumer} convex potential, making use of the expenditure function to characterize competitive equilibrium prices in homothetic Fisher markets.
An important property of this convex potential is that its gradient is equal to the negative excess demand in the market, implying that mirror descent on it is equivalent to \emph{t\^atonnement}, an observation used to prove previous convergence results regarding \emph{t\^atonnement\/} in Fisher markets \cite{fisher-tatonnement}.
Using the bound we derived on the change in the expenditure function as a function of the change in prices, we then showed that the potential function we considered is Bregman-smooth w.r.t.\@ the KL divergence throughout a trajectory of \emph{t\^atonnement}.
Combining this result with the sublinear convergence rate of mirror descent for Bregman-smooth functions \cite{grad-prop-response}, we concluded that \emph{t\^atonnement\/} converges at a rate of $O(\nicefrac{(1+ \elastic^2)}{T})$, where $\elastic$ is the maximum absolute value of the price elasticity of Hicksian demand across all buyers.
Our result not only generalizes existing convergence results for CES and nested CES Fisher markets, but extends them beyond Fisher markets with concave utility functions.
Our convergence rate covers the full spectrum of (nested) CES utilities, including Leontief and linear utilities, unifying previously existing disparate convergence and non-convergence results.
In particular, for $\elastic = 0$, i.e., Leontief markets, we recover the best-known convergence rate of $O(\nicefrac{1}{T})$ \cite{fisher-tatonnement}, and as $\elastic \to \infty$, i.e., linear Fisher markets, we obtain the non-convergent behavior of \emph{t\^atonnement}.

Future work could investigate the space of homogeneous utility functions with negative cross-price elasticity of Hicksian demand to possibly derive faster convergence rates than those provided in this paper.
Additionally, it remains to be seen if the bound we have provided in this paper is tight; the greatest lower bound known for the convergence of \emph{t\^atonnement\/} in homothetic Fisher markets is $O(\nicefrac{1}{t^2})$ for Leontief markets \cite{cheung2012tatonnement}, leaving space for improvement.
Finally, \Cref{lemma:homo_elasticity} suggests that to extend convergence results for \emph{t\^atonnement\/} beyond homothetic domains, one might have to consider the Hicksian demand elasticity \emph{w.r.t.\ utility level} rather than price.

%% file: appendix/proofs.tex
\section{Ommited Results and Proofs}\label{sec_ap:proofs}

We start by presenting the first lemma, which shows that the utility level elasticity of Hicksian demand is equal to 1 in homothetic Fisher markets.

\lemmahomoelasticity*
\begin{proof}[Proof of \Cref{lemma:homo_elasticity}]
Recall from \citeauthor{goktas2022consumer} \cite{goktas2022consumer} that for homogeneous utility functions, the Hicksian demand is homogeneous in $\goalutil$, i.e., for all $\lambda \geq 0$, $\hicksian[\buyer](\price, \lambda\goalutil) = \lambda\hicksian[\buyer](\price, \goalutil)$.
Hence, we have:
\begin{align}
    \elastic[{\hicksian[\buyer][\good]}][{\price[k]}](\price, \goalutil[\buyer]) &= \subdiff[{\price[k]}] \hicksian[\buyer][\good](\price, \goalutil[\buyer]) \frac{\price[k]}{\hicksian[\buyer][\good](\price, \goalutil[\buyer])}\\
    &= \goalutil[\buyer] \subdiff[{\price[k]}]  \hicksian[\buyer][\good](\price, 1) \frac{\price[k]}{\goalutil[\buyer] \hicksian[\buyer][\good](\price, 1)} && \text{(Homogeneity of Hicksian demand)}\\
    &= \frac{\price[k]}{\hicksian[\buyer][\good](\price, 1)}\subdiff[{\price[k]}]  \hicksian[\buyer][\good](\price, 1)\\
    &= \frac{\price[k]}{\hicksian[\buyer][\good](\price, 1)}\subdiff[{\price[k]}]  \hicksian[\buyer][\good](\price, 1)\label{eq:hicksian_income}\\
    &= \elastic[{\hicksian[\buyer][\good]}][{\price[k]}](\price, 1)
\end{align}
Additionally, looking back at \Cref{eq:hicksian_income}, since Hicksian demand is homogeneous of degree 1 for homogeneous utility function (see Lemma 5 of \citet{goktas2022consumer}), by Euler's theorem for homogeneous functions \cite{border2017euler}, we have:
$
    \frac{\price[k]}{\hicksian[\buyer][\good](\price, 1)}\subdiff[{\price[k]}]  \hicksian[\buyer][\good](\price, 1) = \frac{\hicksian[\buyer][\good](\price, 1)}{\hicksian[\buyer][\good](\price, 1)} = 1 \enspace .
$
\end{proof}

We recall Shephard's lemma which was used in \Cref{eq:expend_change}:
\begin{restatable}[Shephard's lemma \cite{duality-economics, shephard, generalized-shephard}]{lemma}{shepherd}\label{shephard}
Let $\expend[\buyer](\price, \goalutil[\buyer])$ be the expenditure function of buyer $\buyer$ and $\hicksian[\buyer](\price, \goalutil[\buyer])$ be the Hicksian demand set of buyer $\buyer$.
The subdifferential $\subdiff[{\price}] \expend[\buyer](\price, \goalutil[\buyer])$ is the Hicksian demand at prices $\price$ and utility level $\goalutil[\buyer]$, i.e., $\subdiff[\price] \expend[\buyer](\price, \goalutil[\buyer]) = \hicksian[\buyer](\price, \goalutil[\buyer])$.
\end{restatable}

We first prove that by setting $\gamma$ to be 5 times the maximum demand for any good throughout the 
entropic t\^atonnement process
, we can bound the change in the prices of goods in each round. We will use the fact that the change in the price of each good is bounded as an assumption in most of the following results.

\begin{lemma}
\label{price-change}
Suppose that entropic t\^atonnement process 
is run for all $\iter \in \iters \subseteq \N_+$ with $\gamma \geq 5 \max\limits_{\substack{t \in \iters \\ \good \in \goods}} \{1, \demand[\good][\iter]\}$ and let $\pricediff = \price[][\iter+1] - \price[][\iter]$. then the following holds for all $t \in \N$:
\begin{align*}
    e^{-\frac{1}{5}} \price[\good][\iter]\leq \price[\good][\iter+1] \leq e^{\frac{1}{5}} \price[\good][\iter] \text{ and }\frac{|\pricediff[\good]|}{\price[\good][\iter]} \leq \frac{1}{4}
\end{align*}
\end{lemma}
 
\begin{proof}[Proof of \Cref{price-change}]
The price of of a good $\good \in \goods$ can at most increase by a factor of $e^{\frac{1}{5}}$:
\begin{align*}
    \price[\good][\iter+1] &= \price[\good][\iter]\exp\left\{\frac{\excess[\good](\price[][\iter])}{\gamma}\right\} = \price[\good][\iter] \exp\left\{\frac{\demand[\good][\iter] - 1}{\gamma}\right\} \leq \price[\good][\iter] \exp\left\{\frac{\demand[\good][\iter]}{\gamma}\right\} \leq \price[\good][\iter] \exp\left\{\frac{\demand[\good][\iter]}{5 \max\limits_{\substack{t \in \N\\ \good \in \goods}} \{1, \demand[\good][\iter]\}}\right\}    
    \leq \price[\good][\iter] e^{\frac{1}{5}}
\end{align*}
and decrease by a factor of $e^{-\frac{1}{5}}$:
\begin{align*}
    \price[\good][\iter+1] &= \price[\good][\iter]\exp\left\{\frac{\excess[\good](\price[][\iter])}{\gamma}\right\} = \price[\good][\iter] \exp\left\{\frac{\demand[\good][\iter] - 1}{\gamma}\right\} \geq \price[\good][\iter] \exp\left\{\frac{- 1}{\gamma}\right\} \geq \price[\good][\iter] \exp\left\{\frac{-1}{5 \max\limits_{\substack{t \in \N \\ \good \in \goods}} \{1, \demand[\good][\iter]\}}\right\} \geq \price[\good][\iter] e^{-\frac{1}{5}}
\end{align*}

Hence, we have $e^{-\frac{1}{5}} \price[\good][\iter]\leq \price[\good][\iter+1] \leq e^{\frac{1}{5}} \price[\good][\iter]$. substracting $\price[\good][\iter]$ from both sides and dividing by $\price[\good][\iter]$, we obtain:
\begin{align*}
    \frac{|\pricediff[\good]|}{\price[\good][\iter]} = \frac{|\price[\good][\iter+1] - \price[\good][\iter]|}{\price[\good][\iter]} \leq e^{\nicefrac{1}{5}} - 1 \leq \frac{1}{4} 
\end{align*}
\end{proof}

The following two results are due to \citeauthor{fisher-tatonnement} \cite{fisher-tatonnement}. We include their proofs for completeness. They allows us to relate the change in prices to the KL-divergence.

\begin{lemma}[\citeauthor{fisher-tatonnement}\cite{fisher-tatonnement}]\label{kl-divergence-1}
Fix $\iter \in \N_+$ and let $\pricediff = \price[][\iter+1] - \price[][\iter]$. Suppose that for all $\good \in \goods$, $\frac{|\pricediff[\good]|}{\price[\good][\iter]} \leq \frac{1}{4}$, then:
\begin{align}
    \frac{(\pricediff[\good])^2}{\price[\good][\iter]} \leq \frac{9}{2} \divergence[\mathrm{KL}][{\price[\good][\iter] + \pricediff[\good]}][{\price[\good][\iter]}]
\end{align}
\end{lemma}

\begin{proof}[Proof of \Cref{kl-divergence-1}]
The bound $\log(x) \geq x - x^2$ for $|x| \leq \frac{1}{4}$ is used below:
\begin{align*}
    \divergence[\mathrm{KL}][{\price[\good][\iter] + \pricediff[\good]}][{\price[\good][\iter]}] &= (\price[\good][\iter] + \pricediff[\good])(\log(\price[\good][\iter] + \pricediff[\good])) - ( \price[\good][\iter] + \pricediff[\good] - \price[\good][\iter]\log(\price[\good]) + \price[\good][\iter] - \log(\price[\good][\iter]) \pricediff[\good]\\
    &= - \pricediff[\good] + (\price[\good][\iter] + \pricediff[\good]) \log \left(1 + \frac{\pricediff[\good]}{\price[\good][\iter]}\right)\\
    &\geq - \pricediff[\good] + (\price[\good][\iter] + \pricediff[\good]) \left( \frac{\pricediff[\good]}{\price[\good][\iter]} - \frac{11}{18} \frac{(\pricediff[\good])^2}{(\price[\good][\iter])^2} \right)\\
    &\geq \frac{7}{18} \frac{(\pricediff[\good])^2}{\price[\good][\iter]} \left( 1 - \frac{11}{7} \frac{\pricediff[\good]}{\price[\good][\iter]}\right)\\
    &= \frac{7}{18} \frac{17}{28} \frac{(\pricediff[\good])^2}{\price[\good][\iter]}\\
    &\geq \frac{2}{9} \frac{(\pricediff[\good])^2}{\price[\good][\iter]}
\end{align*}
\end{proof}

\begin{lemma}\label{lemma:spending_change_squared}
Fix $\iter \in \N_+$ and let $\pricediff = \price[][\iter+1] - \price[][\iter]$. Suppose that $\frac{|\pricediff[\good]|}{\price[\good]} \leq \frac{1}{4}$, then for any $c \in (0,1)$, and $\A \in \R^{\numbuyers \times \numgoods}$, and for all $\good \in \goods$:
\begin{align*}
    \frac{1}{\budget[\buyer] }   \sum_{\good \in \goods} \sum_{k \in \goods} a_{il} \marshallian[\buyer][k](\price[][\iter] + c \pricediff, \budget[\buyer]) |\pricediff[\good]||\pricediff[k]| \leq \frac{4}{3} \sum_{l \in \goods} \frac{a_{il}}{\price[l][\iter]}(\pricediff[l])^2
\end{align*}
\end{lemma}

\begin{proof}[Proof of \Cref{lemma:spending_change_squared}]
First, note that since by our assumption the utilities are locally non-satiated, Walras' law is satisfied, i.e., we have $\budget[\buyer] = \sum_{k \in \goods} \marshallian[\buyer][k](\price[][\iter] + c \pricediff, \budget[\buyer]) (\price[k][\iter] + c \pricediff[k])$;
\begin{align*}
    &\budget[\buyer]\sum_{l \in \goods} \frac{a_{il}}{\price[l][\iter]}(\pricediff[l])^2\\ &= \sum_{l \in \goods} \frac{\left(\sum_{k \in \goods} \marshallian[\buyer][k](\price[][\iter] + c \pricediff, \budget[\buyer]) (\price[k][\iter] + c \pricediff[k]) \right)\marshallian[\buyer][l](\price[][\iter] + c \pricediff, \budget[\buyer])}{\price[l][\iter]}(\pricediff[l])^2\\
    &\geq \sum_{l \in \goods} \frac{\left(\sum_{k \in \goods} \marshallian[\buyer][k](\price[][\iter] + c \pricediff, \budget[\buyer]) (\price[k][\iter] - \frac{1}{4}\price[k][\iter]) \right) a_{il}}{\price[l][\iter]}(\pricediff[l])^2\\
    &= \sum_{l \in \goods} \frac{\left(\sum_{k \in \goods} \marshallian[\buyer][k](\price[][\iter] + c \pricediff, \budget[\buyer]) ( \frac{3}{4} \price[k][\iter] ) \right) a_{il}}{\price[l][\iter]}(\pricediff[l])^2\\
    &=\frac{3}{4}\sum_{l \in \goods} \sum_{k \in \goods} a_{il} \marshallian[\buyer][l](\price[][\iter] + c \pricediff, \budget[\buyer])  \frac{\price[k][\iter]}{\price[l][\iter]} (\pricediff[l])^2\\
    &= \frac{3}{4} \left[\sum_{l \in \goods} a_{il} \marshallian[\buyer][l](\price[][\iter] + c \pricediff, \budget[\buyer]) (\pricediff[l])^2 + \sum_{l \in \goods} \sum_{k \neq l } a_{il} \marshallian[\buyer][k](\price[][\iter] + c \pricediff, \budget[\buyer]) \frac{\price[k][\iter]}{\price[l][\iter]} (\pricediff[l])^2 \right]\\
    &= \frac{3}{4} \left[ \sum_{l \in \goods} a_{il} \marshallian[\buyer][l](\price[][\iter] + c \pricediff, \budget[\buyer]) (\pricediff[l])^2 + \sum_{k \in \goods} \sum_{k \leq l} a_{il} \marshallian[\buyer][k](\price[][\iter] + c \pricediff, \budget[\buyer]) \left( \frac{\price[k][\iter]}{\price[l][\iter]} |\pricediff[l]|^2 + \frac{\price[l][\iter]}{\price[k][\iter]}  |\pricediff[k]|^2\right) \right]
\end{align*}

Now, we apply the AM-GM inequality, i.e., for all $x, y \in \R_+$ since $\sqrt{xy} \leq \frac{x+y}{2}$, we have: 
\begin{align*}
\budget[\buyer] \sum_{l \in \goods} \frac{a_{il}}{\price[l][\iter]} (\pricediff[l])^2 &\geq \frac{3}{4} \sum_{l \in \goods} a_{il} \marshallian[\buyer][l](\price[][\iter] + c \pricediff, \budget[\buyer]) (\pricediff[l])^2 + \sum_{k < l} a_{il} \marshallian[\buyer][k](\price[][\iter] + c \pricediff, \budget[\buyer])\left(2 |\pricediff[l]||\pricediff[k]||\right)\\
&= \frac{3}{4} \sum_{\good \in \goods} \sum_{k \in \goods} a_{il} \marshallian[\buyer][k](\price[][\iter] + c \pricediff, \budget[\buyer]) |\pricediff[\good]||\pricediff[k]|
\end{align*}
\end{proof}

\begin{lemma}{\cite{fisher-tatonnement}}\label{leontief-ineq}
For all $\good \in \goods$:
\begin{align*}
    \frac{1}{\budget[\buyer] } \sum_{\good \in \goods} \sum_{k \in \goods} \marshallian[\buyer][\good][\iter]\marshallian[\buyer][k][\iter]|\pricediff[\good]||\pricediff[k]| \leq \sum_{l \in \goods} \frac{\marshallian[\buyer][l][\iter]}{\price[l][\iter]}(\pricediff[l])^2
\end{align*}
\end{lemma}

\begin{proof}[Proof of \Cref{leontief-ineq}]
First, note that by Walras' law we have $\budget[\buyer] = \sum_{k \in \goods} \marshallian[\buyer][k][\iter] \price[k][\iter]$;
\begin{align*}
    \budget[\buyer]\sum_{l \in \goods} \frac{\marshallian[\buyer][l][\iter]}{\price[l][\iter]}(\pricediff[l])^2 &= \sum_{l \in \goods} \frac{\left(\sum_{k \in \goods} \marshallian[\buyer][k][\iter] \price[k][\iter] \right)\marshallian[\buyer][l][\iter]}{\price[l][\iter]}(\pricediff[l])^2\\
    &=\sum_{l \in \goods} \sum_{k \in \goods} \marshallian[\buyer][l][\iter] \marshallian[\buyer][k][\iter] \frac{\price[k][\iter]}{\price[l][\iter]} (\pricediff[l])^2\\
    &= \sum_{l \in \goods} (\marshallian[\buyer][l][\iter])^2 (\pricediff[l])^2 + \sum_{l \in \goods} \sum_{k \neq l } \marshallian[\buyer][l][\iter] \marshallian[\buyer][k][\iter] \frac{\price[k][\iter]}{\price[l][\iter]} (\pricediff[l])^2\\
    &= \sum_{l \in \goods} (\marshallian[\buyer][l][\iter])^2 (\pricediff[l])^2 + \sum_{k \in \goods} \sum_{k \leq l} \marshallian[\buyer][k][\iter] \marshallian[\buyer][l][\iter] \left( \frac{\price[k][\iter]}{\price[l][\iter]} (\pricediff[l])^2 + \frac{\price[l][\iter]}{\price[k][\iter]} \ (\pricediff[k])^2\right)
\end{align*}

Now, we apply the AM-GM inequality: 
\begin{align*}
\budget[\buyer] \sum_{l \in \goods} \frac{\marshallian[\buyer][l][\iter]}{\price[l][\iter]} (\pricediff[l])^2 &\geq \sum_{l \in \goods} (\marshallian[\buyer][l][\iter])^2 (\pricediff[l])^2 + \sum_{k < l} \marshallian[\buyer][k][\iter] \marshallian[\buyer][l][\iter]\left(2 |\pricediff[l]||\pricediff[k]||\right)\\
&= \sum_{\good \in \goods} \sum_{k \in \goods} \marshallian[\buyer][\good][\iter] \marshallian[\buyer][k][\iter] |\pricediff[\good]||\pricediff[k]|
\end{align*}
\end{proof}

An important result in microeconomics is the \mydef{law of demand} which states that when the price of a good increases, the Hicksian demand for that good decreases in a very general setting of utility functions \cite{levin-notes, mas-colell}. We state a weaker version of the law of demand which is re-formulated to fit the t\^atonnement framework.

\begin{lemma}[Law of Demand]\cite{levin-notes,mas-colell}\label{law-of-demand}
Suppose that $\forall \good \in \goods, t \in \N, \price[\good][\iter], \price[\good][\iter+1] \geq 0$ and $\util[\buyer]$ is continuous and concave. Then, $\sum_{\good \in \goods} \pricediff[\good] \left(\hicksian[\buyer][\good][\iter+1] - \hicksian[\buyer][\good][t]\right) \leq 0$.
\end{lemma}

A simple corollary of the law of demand which is used throughout the rest of this paper is that, during t\^atonnement, the change in expenditure of the next time period is always less than or equal to the change in expenditure of the previous time period's.

\begin{corollary}\label{law-of-demand-corollary}
Suppose that $\forall t \in \N, \good \in \goods, \price[\good][\iter],  \price[\good][\iter+1] \geq 0$ and $\util[\buyer]$ is continuous and concave then $\forall t \in \N, \sum_{\good \in \goods} \pricediff[\good] \hicksian[\buyer][\good][\iter+1]  \leq \sum_{\good \in \goods} \pricediff[\good] \hicksian[\buyer][\good][\iter]$.
\end{corollary} 

The following lemma simply restates an essential fact about expenditure functions and Hicksian demand, namely that the Hicksian demand is the minimizer of the expenditure function.

\begin{lemma}\label{hicksian-is-expend-minimizer}
Suppose that $\forall \good \in \goods, \price[\good][\iter], \iter \in \N, \hicksian[\buyer][\good][\iter], \hicksian[\buyer][\good][\iter+1] \geq 0$ and $\util[\buyer]$ is continuous and concave then $\sum_{\good \in \goods} \hicksian[\buyer][\good][\iter] \price[\good][\iter] \leq \sum_{\good \in \goods} \hicksian[\buyer][\good][\iter+1] \price[\good][\iter]$.
\end{lemma}

\begin{proof}[\Cref{hicksian-is-expend-minimizer}]
For the sake of contradiction, assume that $\sum_{\good \in \goods} \hicksian[\buyer][\good][\iter] \price[\good][\iter] > \sum_{\good \in \goods} \hicksian[\buyer][\good][\iter+1] \price[\good][\iter]$. By the definition of the Hicksian demand, we know that the bundle $\hicksian[\buyer][][\iter]$ provides the buyer with one unit of utility. 
Recall that the expenditure at any price $\price$ is equal to the sum of the product of the Hicksian demands and prices, that is $\expend[\buyer](\price, 1) = \sum_{\good \in \goods} \hicksian[\buyer][\good](\price, 1) \price[\good]$. Hence, we have $\expend[\buyer](\price[\good][\iter], 1) = \sum_{\good \in \goods} \hicksian[\buyer][\good][\iter] \price[\good][\iter] > \sum_{\good \in \goods} \hicksian[\buyer][\good][\iter+1] \price[\good][\iter] = \expend[\buyer](\price[][\iter], 1)$, a contradiction.
\end{proof}

We now introduce the following lemma which makes use of results on the behavior of Hicksian demand and expenditure functions in homothetic Fisher markets introduced by \citeauthor{goktas2022consumer} \cite{goktas2022consumer}. In conjunction with \Cref{law-of-demand-corollary} and \Cref{hicksian-is-expend-minimizer} are key in proving that \Cref{ineq-devanur} holds allowing us to establish convergence of t\^atonnement in a general setting of utility functions.
Additionally, the lemma relates the Marshallian demand of homogeneous utility functions to their Hicksian demand. Before we present the lemma, we recall the following identities \cite{mas-colell}:
\begin{align}
    &\forall \budget[\buyer] \in \mathbb{R}_{+} & \expend[\buyer](\price, \indirectutil[\buyer](\price, \budget[\buyer])) = \budget[\buyer]\label{expend-to-budget}\\
    &\forall \goalutil[\buyer] \in \mathbb{R}_{+} &  \indirectutil[\buyer](\price, \expend[\buyer](\price, \goalutil[\buyer])) = \goalutil[\buyer]\label{indirect-to-value}\\
    & \forall \budget[\buyer] \in \mathbb{R}_{+} & \hicksian[\buyer](\price, \indirectutil[\buyer](\price, \budget[\buyer])) = \marshallian[\buyer](\price, \budget[\buyer])
    \label{hicksian-marshallian}\\
    &\forall \goalutil[\buyer] \in \mathbb{R}_{+} &  \marshallian[\buyer](\price, \expend[\buyer](\price, \goalutil[\buyer])) = \hicksian[\buyer](\price, \goalutil[\buyer])\label{marshallian-hicksian}
\end{align}

\begin{lemma}\label{equiv-def-demand}
Suppose that $\util[\buyer]$ is continuous and homogeneous, then the following holds:
\begin{align*}
    &\forall \good \in \goods &\marshallian[\buyer][\good](\price, \budget[\buyer]) = \frac{\budget[\buyer] \hicksian[\buyer][\good](\price, 1)}{\sum_{\good \in \goods} \hicksian[\buyer][\good](\price, 1) \price[\good]}
\end{align*}
\end{lemma}

\begin{proof}[Proof of \Cref{equiv-def-demand}]
We note that when utility function $\util[\buyer]$ is strictly concave, the Marshallian and Hicksian demand are unique making the following equalities well-defined.
\begin{align*}
    \frac{\budget[\buyer] \hicksian[\buyer][\good](\price, 1)}{\sum_{\good \in \goods} \hicksian[\buyer][\good](\price, 1) \price[\good]} &= \frac{\budget[\buyer] \hicksian[\buyer][\good](\price, 1)}{\expend[\buyer](\price, 1)} && \text{(Definition of expenditure function)}\\
    &= \budget[\buyer]\indirectutil[\buyer](\price, 1) \hicksian[\buyer][\good](\price, 1) && \text{(Corollary 1 of \citeauthor{goktas2022consumer} \cite{goktas2022consumer})} 
    \\
    &= \indirectutil[\buyer](\price, \budget[\buyer]) \hicksian[\buyer][\good](\price, 1)
    \\
    &= \hicksian[\buyer][\good](\price, \indirectutil[\buyer](\price, \budget[\buyer])) 
    \\
    &= \marshallian[\buyer][\good](\price, \budget[\buyer]) && \text{(Marshallian Demand Identity \Cref{hicksian-marshallian})}
\end{align*}
\end{proof}

The following lemma proves that the relative change in expenditures at each iteration of tatonnement is bounded when the relative change in prices is bounded.
\begin{lemma}\label{lemma:expend_relative_change}
Suppose that $\forall \good \in \goods, \frac{|\pricediff[\good]|}{\price[\good][\iter]} \leq \frac{1}{4}$, then for any $t \in \N_+$ and $\buyer \in \buyers$:
\begin{align}
    \left|\frac{\left< \hicksian[\buyer](\price[][\iter+1], 1), \price[][\iter+1]\right> - \left< \hicksian[\buyer](\price[][\iter], 1), \price[][\iter]\right> }{\left<\hicksian[\buyer][][\iter], \price[][\iter] \right>} \right| \leq \frac{1}{4}
\end{align}
\end{lemma}
\begin{proof}[Proof of \Cref{lemma:expend_relative_change}]
\noindent Case 1: $\left< \hicksian[\buyer](\price[][\iter+1], 1), \price[][\iter+1]\right> \geq \left< \hicksian[\buyer](\price[][\iter], 1), \price[][\iter]\right>$ 
    \begin{align}
        &\frac{\left< \hicksian[\buyer](\price[][\iter+1], 1), \price[][\iter+1]\right> - \left< \hicksian[\buyer](\price[][\iter], 1), \price[][\iter]\right> }{\left<\hicksian[\buyer][][\iter], \price[][\iter] \right>} \\
        &\leq \frac{\left< \hicksian[\buyer](\price[][\iter], 1), \price[][\iter+1]\right> - \left< \hicksian[\buyer](\price[][\iter], 1), \price[][\iter]\right> }{\left<\hicksian[\buyer][][\iter], \price[][\iter] \right>}  && \text{(\Cref{law-of-demand-corollary})}\\
        &= \frac{\left< \hicksian[\buyer](\price[][\iter], 1), \price[][\iter+1] \right>}{\left<\hicksian[\buyer][][\iter], \price[][\iter] \right>} - 1  \\
        &\leq \frac{5}{4}\frac{\left< \hicksian[\buyer](\price[][\iter], 1), \price[][\iter]\right>}{\left<\hicksian[\buyer][][\iter], \price[][\iter] \right>} - 1  \\
        &= \frac{1}{4}
    \end{align}
\noindent 
where the penultimate line follows from the assumption that $\forall \good \in \goods, \frac{|\pricediff[\good]|}{\price[\good][\iter]} \leq \frac{1}{4}$.

\noindent Case 2: $\left< \hicksian[\buyer](\price[][\iter+1], 1), \price[][\iter+1]\right> \leq \left< \hicksian[\buyer](\price[][\iter], 1), \price[][\iter]\right>$
\begin{align}
    \frac{\left< \hicksian[\buyer](\price[][\iter], 1), \price[][\iter]\right> - \left< \hicksian[\buyer](\price[][\iter+1], 1), \price[][\iter+1]\right> }{\left<\hicksian[\buyer][][\iter], \price[][\iter] \right>} &= 1 - \frac{\left< \hicksian[\buyer](\price[][\iter+1], 1), \price[][\iter+1]\right> }{\left<\hicksian[\buyer][][\iter], \price[][\iter] \right>}\\
    &\leq 1 - \frac{3}{4} \frac{\left< \hicksian[\buyer](\price[][\iter+1], 1), \price[][\iter]\right> }{\left<\hicksian[\buyer][][\iter], \price[][\iter] \right>}\\
    &\leq 1 - \frac{3}{4} \frac{\left< \hicksian[\buyer](\price[][\iter], 1), \price[][\iter]\right> }{\left<\hicksian[\buyer][][\iter], \price[][\iter] \right>} && \text{(\Cref{law-of-demand-corollary})}\\
    &= \frac{1}{4}
\end{align}
where the second line follows from the assumption that $\forall \good \in \goods, \frac{|\pricediff[\good]|}{\price[\good][\iter]} \leq \frac{1}{4}$.
\end{proof}

\begin{lemma}\label{lemma:expend_change_squared}
Suppose that for all $\good \in \goods$, $\frac{|\pricediff[\good]|}{\price[\good]} \leq \frac{1}{4}$, then for some $c \in (0, 1)$ and $t \in \N_+$, we have:
\begin{align}
    &\frac{1}{\budget[\buyer]} \left( \left<\marshallian[\buyer][][t], \pricediff \right> +\frac{\budget[\buyer]}{2} \frac{\left< \grad[\price]^2 \expend[\buyer](\price[][\iter] + c \pricediff, 1) \pricediff, \pricediff \right>}{\expend[\buyer](\price[][\iter], 1)} \right)^2\\ 
    &\leq \left(1+ \frac{5 \elastic}{9} \right)\sum_{l \in \goods} \frac{\marshallian[\buyer][l][\iter]}{\price[l][\iter]}(\pricediff[l])^2  + \left(\frac{25\elastic^2 }{432} \right) \sum_{l \in \goods} \frac{\marshallian[\buyer][l](\price[][\iter] + c \pricediff, \budget[\buyer])}{\price[l][\iter]}(\pricediff[l])^2
\end{align}
\end{lemma}

\begin{proof}[Proof of \Cref{lemma:expend_change_squared}]
\begin{align}
     &\frac{1}{\budget[\buyer]} \left( \left<\marshallian[\buyer][][t], \pricediff \right> +\frac{\budget[\buyer]}{2} \frac{\left< \grad[\price]^2 \expend[\buyer](\price[][\iter] + c \pricediff, 1) \pricediff, \pricediff \right>}{\expend[\buyer](\price[][\iter], 1)} \right)^2 \\
     &= \frac{1}{\budget[\buyer]} \left[  \left<\marshallian[\buyer][][t], \pricediff \right>^2 + 2 \left<\marshallian[\buyer][][t], \pricediff \right> \left(\frac{\budget[\buyer]}{2} \frac{\left< \grad[\price]^2 \expend[\buyer](\price[][\iter] + c \pricediff, 1) \pricediff, \pricediff \right>}{\expend[\buyer](\price[][\iter], 1)} \right)  + \left(\frac{\budget[\buyer]}{2} \frac{\left< \grad[\price]^2 \expend[\buyer](\price[][\iter] + c \pricediff, 1) \pricediff, \pricediff \right>}{\expend[\buyer](\price[][\iter], 1)}  \right)^2  \right]\\
     &\leq \frac{1}{\budget[\buyer]} \left[  \left| \left<\marshallian[\buyer][][t], \pricediff \right>^2 \right| + 2 \left|\left<\marshallian[\buyer][][t], \pricediff \right> \right| \left|\frac{\budget[\buyer]}{2} \frac{\left< \grad[\price]^2 \expend[\buyer](\price[][\iter] + c \pricediff, 1) \pricediff, \pricediff \right>}{\expend[\buyer](\price[][\iter], 1)} \right|  + \left|\frac{\budget[\buyer]}{2} \frac{\left< \grad[\price]^2 \expend[\buyer](\price[][\iter] + c \pricediff, 1) \pricediff, \pricediff \right>}{\expend[\buyer](\price[][\iter], 1)}  \right|^2  \right]\\
     &\leq \frac{1}{\budget[\buyer]} \left[ \left<\marshallian[\buyer][][t], \left|\pricediff\right| \right>^2  + 2 \left<\marshallian[\buyer][][t], \left|\pricediff\right| \right>  \left|\frac{\budget[\buyer]}{2} \frac{\left< \grad[\price]^2 \expend[\buyer](\price[][\iter] + c \pricediff, 1) \pricediff, \pricediff \right>}{\expend[\buyer](\price[][\iter], 1)} \right|  + \left|\frac{\budget[\buyer]}{2} \frac{\left< \grad[\price]^2 \expend[\buyer](\price[][\iter] + c \pricediff, 1) \pricediff, \pricediff \right>}{\expend[\buyer](\price[][\iter], 1)}  \right|^2  \right]
\end{align}
\noindent where we denote $\left|\pricediff \right| = (\left|\pricediff[1] \right|, \hdots, \left|\pricediff[\numgoods] \right|) $.

\begin{align}
    &\leq \frac{1}{\budget[\buyer]} \left[ \left<\marshallian[\buyer][][t], \left|\pricediff\right| \right>^2  + 2 \left<\marshallian[\buyer][][t], \left|\pricediff \right| \right> \notag \right. \\ 
    &\left. \left( \frac{5\elastic  }{6}    \sum_{\good} \frac{(\pricediff[\good])^2}{\price[\good]}  \marshallian[\buyer][\good](\price[][\iter] + c \pricediff, \budget[\buyer]) \right) 
      +\left(\frac{5\elastic  }{6}  \sum_{\good} \frac{(\pricediff[\good])^2}{\price[\good]}  \marshallian[\buyer][\good](\price[][\iter] + c \pricediff, \budget[\buyer])  \right)^2  \right]&& \text{(\Cref{lemma:expend_change})}\\
     &= \frac{1}{\budget[\buyer]} \left[ \left<\marshallian[\buyer][][t], \left|\pricediff\right| \right>^2  + \frac{5\elastic }{3}  \left<\marshallian[\buyer][][t], \left|\pricediff\right| \right>  \left(   \sum_{\good} \frac{(\pricediff[\good])^2}{\price[\good]}  \marshallian[\buyer][\good](\price[][\iter] + c \pricediff, \budget[\buyer]) \right) 
     \right. \notag \\ &+ \left.
     \frac{25\elastic^2 }{36}   \left(  \sum_{\good} \frac{(\pricediff[\good])^2}{\price[\good]}  \marshallian[\buyer][\good](\price[][\iter] + c \pricediff, \budget[\buyer])  \right)^2  \right]
\end{align}

Since $\forall \good \in \goods, \frac{|\pricediff[\good]|}{\price[\good][\iter]} \leq \frac{1}{4}$, we have: 
\begin{align}
    &\leq\frac{1}{\budget[\buyer]} \left[ \left<\marshallian[\buyer][][t], \left|\pricediff \right|\right>^2  + \frac{5 \elastic }{12}  \left<\marshallian[\buyer][][t], \left|\pricediff\right| \right>  \left(   \sum_{\good} \left|\pricediff[\good]\right| \marshallian[\buyer][\good](\price[][\iter] + c \pricediff, 1) \right) 
    +\frac{25\elastic^2 }{576} \left(  \sum_{\good} |\pricediff[\good]| \marshallian[\buyer][\good](\price[][\iter] + c \pricediff, \budget[\buyer])  \right)^2  \right]\\
     &= \frac{1}{\budget[\buyer]}  \sum_{\good \in \goods} \sum_{k \in \goods} \marshallian[\buyer][\good][\iter] \marshallian[\buyer][k][\iter] |\pricediff[\good]| |\pricediff[k]|   + \frac{1}{\budget[\buyer]} \frac{5 \elastic }{12} \sum_{\good} \sum_{k} \marshallian[\buyer][k][\iter]  \marshallian[\buyer][\good](\price[][\iter] + c \pricediff, \budget[\buyer]) \left|\pricediff[k] \right|     \left|\pricediff[\good]\right|
      \notag \\ &+ 
    \frac{1}{\budget[\buyer]} \frac{25\elastic^2 }{576}   \sum_{\good} \sum_{k} \marshallian[\buyer][\good](\price[][\iter] + c \pricediff, \budget[\buyer])  \marshallian[\buyer][k](\price[][\iter] + c \pricediff, \budget[\buyer])  |\pricediff[k]| |\pricediff[\good]|  \\
    &\leq  \sum_{l \in \goods} \frac{\marshallian[\buyer][l][\iter]}{\price[l][\iter]}(\pricediff[l])^2  + \frac{1}{\budget[\buyer]} \frac{5 \elastic }{12} \sum_{\good} \sum_{k} \marshallian[\buyer][k][\iter]  \marshallian[\buyer][\good](\price[][\iter] + c \pricediff, \budget[\buyer]) \left|\pricediff[k] \right|     \left|\pricediff[\good]\right|
      \notag \\ &+ 
    \frac{1}{\budget[\buyer]} \frac{25\elastic^2 }{576}   \sum_{\good} \sum_{k} \marshallian[\buyer][\good](\price[][\iter] + c \pricediff, \budget[\buyer])  \marshallian[\buyer][k](\price[][\iter] + c \pricediff, \budget[\buyer])  |\pricediff[k]| |\pricediff[\good]| 
\end{align}
where the last line was obtained by (\Cref{leontief-ineq}). Continuing, by \Cref{lemma:spending_change_squared}, we have:
\begin{align}
    &\leq  \sum_{l \in \goods} \frac{\marshallian[\buyer][l][\iter]}{\price[l][\iter]}(\pricediff[l])^2  + \frac{5 \elastic }{12}  \frac{4}{3} \sum_{l \in \goods} \frac{\marshallian[\buyer][l][\iter]}{\price[l][\iter]}(\pricediff[l])^2 + \frac{25\elastic^2 }{576}  \frac{4}{3} \sum_{l \in \goods} \frac{\marshallian[\buyer][l](\price[][\iter] + c \pricediff, \budget[\buyer])}{\price[l][\iter]}(\pricediff[l])^2  \\
    &= \left(1+ \frac{5 \elastic}{9} \right)\sum_{l \in \goods} \frac{\marshallian[\buyer][l][\iter]}{\price[l][\iter]}(\pricediff[l])^2  + \left(\frac{25\elastic^2 }{432} \right) \sum_{l \in \goods} \frac{\marshallian[\buyer][l](\price[][\iter] + c \pricediff, \budget[\buyer])}{\price[l][\iter]}(\pricediff[l])^2
\end{align}

\end{proof}

\begin{lemma}\label{fix-buyer-bound}
Suppose that $\frac{|\pricediff[\good]|}{\price[\good][\iter]} \leq \frac{1}{4}$, then
\begin{align*}
    &\budget[\buyer] \log \left( 1 - \frac{\expend[\buyer](\price[][\iter+1], 1) - \expend[\buyer](\price[][\iter], 1)}{\expend[\buyer](\price[][\iter], 1)} \left(1+ \frac{\expend[\buyer](\price[][\iter+1], 1) - \expend[\buyer](\price[][\iter], 1)}{\expend[\buyer](\price[][\iter], 1)} \right)^{-1} \right) \\
    &\leq \left(\frac{4}{3}+\frac{20\elastic}{27}\right) \sum_{l \in \goods} \frac{\marshallian[\buyer][l][\iter]}{\price[l][\iter]}(\pricediff[l])^2  + \left(\frac{5 \elastic}{6} + \frac{25\elastic^2 }{324} \right) \sum_{l \in \goods} \frac{(\pricediff[l])^2}{\price[l][\iter]}\marshallian[\buyer][l](\price[][\iter] + c \pricediff, \budget[\buyer]) - \left< \marshallian[\buyer](\price[][\iter], 1), \pricediff \right>
\end{align*}
\end{lemma}

\begin{proof}[Proof of \Cref{fix-buyer-bound}]
First, we note that $ \hicksian[\buyer][][\iter] \cdot \price[][\iter] > 0$ because prices during our t\^atonnement rule reach 0 only asymptotically and Hicksian demand for one unit of utility at prices $\price[][\iter] > 0$ is strictly positive; and likewise, prices reach $\infty$ only asymptotically, which implies that Hicksian demand is always strictly positive. This fact will come handy, as we divide some expressions by  $ \hicksian[\buyer][][\iter] \cdot \price[][\iter]$.

Fix $t \in \N_+$ and $\buyer \in \buyers$.
Since by our assumptions $\frac{|\pricediff[\good]|}{\price[\good][\iter]} \leq \frac{1}{4} $, by \Cref{lemma:expend_relative_change}, we have $0 \leq \left|\frac{\expend[\buyer](\price[][\iter+1], 1) - \expend[\buyer](\price[][\iter], 1)}{\expend[\buyer](\price[][\iter], 1)}\right| \leq \frac{1}{4}$. We can then use the bound $1-x(1+x)^{-1} \leq 1 + \frac{4}{3} x^2 - x$, for $0 \leq |x| \leq \frac{1}{4}$, with $x = \frac{\expend[\buyer](\price[][\iter+1], 1) - \expend[\buyer](\price[][\iter], 1)}{\expend[\buyer](\price[][\iter], 1)}$, to get:
\begin{align*}
    \budget[\buyer] \log \left( 1 - \frac{\expend[\buyer](\price[][\iter+1], 1) - \expend[\buyer](\price[][\iter], 1)}{\expend[\buyer](\price[][\iter], 1)} \left(1+ \frac{\expend[\buyer](\price[][\iter+1], 1) - \expend[\buyer](\price[][\iter], 1)}{\expend[\buyer](\price[][\iter], 1)} \right)^{-1} \right) \\
    \leq \budget[\buyer] \log \left( 1 + \frac{4}{3}\left(\frac{\expend[\buyer](\price[][\iter+1], 1) - \expend[\buyer](\price[][\iter], 1)}{\expend[\buyer](\price[][\iter], 1)} \right)^2 - \frac{\expend[\buyer](\price[][\iter+1], 1) - \expend[\buyer](\price[][\iter], 1)}{\expend[\buyer](\price[][\iter], 1)}\right)
\end{align*}

\noindent
Let $a = \frac{4}{3}\left(\frac{\expend[\buyer](\price[][\iter+1], 1) - \expend[\buyer](\price[][\iter], 1)}{\expend[\buyer](\price[][\iter], 1)} \right)^2 - \frac{\expend[\buyer](\price[][\iter+1], 1) - \expend[\buyer](\price[][\iter], 1)}{\expend[\buyer](\price[][\iter], 1)}$. By \Cref{lemma:expend_relative_change}, we know that $0 + (-\nicefrac{1}{4}) \leq a \leq \frac{1}{12} + \nicefrac{1}{4} \Leftrightarrow -\nicefrac{1}{4} \leq a \leq \frac{1}{3}$. We now use the bound $x \geq \log \left( 1 + x \right)$ for $x > -1$, with $x = a$ to get:
\begin{align*}
    &\budget[\buyer] \log \left(1  + \frac{4}{3}\left(\frac{\expend[\buyer](\price[][\iter+1], 1) - \expend[\buyer](\price[][\iter], 1)}{\expend[\buyer](\price[][\iter], 1)} \right)^2 - \frac{\expend[\buyer](\price[][\iter+1], 1) - \expend[\buyer](\price[][\iter], 1)}{\expend[\buyer](\price[][\iter], 1)}\right)\\
    &\leq \budget[\buyer] \left( \frac{4}{3}\left(\frac{\expend[\buyer](\price[][\iter+1], 1) - \expend[\buyer](\price[][\iter], 1)}{\expend[\buyer](\price[][\iter], 1)} \right)^2 - \frac{\expend[\buyer](\price[][\iter+1], 1) - \expend[\buyer](\price[][\iter], 1)}{\expend[\buyer](\price[][\iter], 1)} \right)
\end{align*}

Using a first order Taylor expansion of $\expend[\buyer](\price[][\iter] + \pricediff, 1)$ around $\price[][\iter]$, by Taylor's theorem \cite{graves1927riemann}, we have: $\expend[\buyer](\price[][\iter] + \pricediff, 1) = \expend[\buyer](\price[][\iter], 1) + \left<\grad[\price] \expend[\buyer](\price[][\iter], 1), \pricediff \right> + \nicefrac{1}{2} \left< \grad[\price]^2 \expend[\buyer](\price[][\iter] + c \pricediff, 1) \pricediff, \pricediff \right>$ for some $c \in (0, 1)$. Re-organizing terms around, we get $\expend[\buyer](\price[][\iter+1], 1) - \expend[\buyer](\price[][\iter], 1) =  \left<\grad[\price] \expend[\buyer](\price[][\iter], 1), \pricediff \right> + \nicefrac{1}{2} \left< \grad[\price]^2 \expend[\buyer](\price[][\iter] + c \pricediff, 1) \pricediff, \pricediff \right>$, which gives us:
\begin{align}
    &= \budget[\buyer] \left( \frac{4}{3}\left(\frac{\left<\grad[\price] \expend[\buyer](\price[][\iter], 1), \pricediff \right> + \nicefrac{1}{2} \left< \grad[\price]^2 \expend[\buyer](\price[][\iter] + c \pricediff, 1) \pricediff, \pricediff \right>}{\expend[\buyer](\price[][\iter], 1)} \right)^2 - \right. \notag \\ 
    & \left. \frac{\left<\grad[\price] \expend[\buyer](\price[][\iter], 1), \pricediff \right> + \nicefrac{1}{2} \left< \grad[\price]^2 \expend[\buyer](\price[][\iter] + c \pricediff, 1) \pricediff, \pricediff \right>}{\expend[\buyer](\price[][\iter], 1)} \right)
\end{align}

\noindent
Continuing, by Shepherd's lemma \cite{shephard}, we have:

\begin{align}
    &= \frac{4}{3} \budget[\buyer] \left( \frac{\left< \hicksian[\buyer](\price[][\iter], 1), \pricediff \right> + \nicefrac{1}{2} \left< \grad[\price]^2 \expend[\buyer](\price[][\iter] + c \pricediff, 1) \pricediff, \pricediff \right>}{\expend[\buyer](\price[][\iter], 1)} \right)^2 - \notag \\
     &\budget[\buyer] \frac{\left<\hicksian[\buyer](\price[][\iter], 1), \pricediff \right> + \nicefrac{1}{2} \left< \grad[\price]^2 \expend[\buyer](\price[][\iter] + c \pricediff, 1) \pricediff, \pricediff \right>}{\expend[\buyer](\price[][\iter], 1)}\\
    &= \frac{4}{3} \frac{1}{\budget[\buyer]} \left( \frac{\budget[\buyer] \left< \hicksian[\buyer](\price[][\iter], 1), \pricediff \right>}{\expend[\buyer](\price[][\iter], 1)} + \frac{\nicefrac{\budget[\buyer]}{2} \left< \grad[\price]^2 \expend[\buyer](\price[][\iter] + c \pricediff, 1) \pricediff, \pricediff \right>}{\expend[\buyer](\price[][\iter], 1)} \right)^2 - \notag \\
     &\frac{\budget[\buyer] \left<\hicksian[\buyer](\price[][\iter], 1), \pricediff \right>}{\expend[\buyer](\price[][\iter], 1)} - \frac{\nicefrac{\budget[\buyer]}{2} \left< \grad[\price]^2 \expend[\buyer](\price[][\iter] + c \pricediff, 1) \pricediff, \pricediff \right>}{\expend[\buyer](\price[][\iter], 1)}\\
    &= \frac{4}{3} \frac{1}{\budget[\buyer]}  \left( \left< \marshallian[\buyer](\price[][\iter], 1), \pricediff \right> + \frac{\nicefrac{\budget[\buyer]}{2} \left< \grad[\price]^2 \expend[\buyer](\price[][\iter] + c \pricediff, 1) \pricediff, \pricediff \right>}{\expend[\buyer](\price[][\iter], 1)} \right)^2 - \notag \\
     &\left< \marshallian[\buyer](\price[][\iter], 1), \pricediff \right> - \frac{\nicefrac{\budget[\buyer]}{2} \left< \grad[\price]^2 \expend[\buyer](\price[][\iter] + c \pricediff, 1) \pricediff, \pricediff \right>}{\expend[\buyer](\price[][\iter], 1)}
\end{align}
\noindent where the last line was obtained from \Cref{equiv-def-demand}.

Using \Cref{lemma:expend_change_squared}, we have: 
\begin{align}
    &\leq \frac{4}{3} \left( \left(1+\frac{5\elastic}{9}\right)\sum_{l \in \goods} \frac{\marshallian[\buyer][l][\iter]}{\price[l][\iter]}(\pricediff[l])^2  + \frac{25\elastic^2 }{432}  \sum_{l \in \goods} \frac{(\pricediff[l])^2}{\price[l][\iter]}\marshallian[\buyer][l](\price[][\iter] + c \pricediff, \budget[\buyer]) \right) - \notag \\
     &\left< \marshallian[\buyer](\price[][\iter], 1), \pricediff \right> - \frac{\nicefrac{\budget[\buyer]}{2} \left< \grad[\price]^2 \expend[\buyer](\price[][\iter] + c \pricediff, 1) \pricediff, \pricediff \right>}{\expend[\buyer](\price[][\iter], 1)}\\
    &=\left(\frac{4}{3}+\frac{20\elastic}{27}\right) \sum_{l \in \goods} \frac{\marshallian[\buyer][l][\iter]}{\price[l][\iter]}(\pricediff[l])^2  +\frac{25\elastic^2 }{324} \sum_{l \in \goods} \frac{(\pricediff[l])^2}{\price[l][\iter]}\marshallian[\buyer][l](\price[][\iter] + c \pricediff, \budget[\buyer]) - \notag \\
     &\left< \marshallian[\buyer](\price[][\iter], 1), \pricediff \right> - \frac{\nicefrac{\budget[\buyer]}{2} \left< \grad[\price]^2 \expend[\buyer](\price[][\iter] + c \pricediff, 1) \pricediff, \pricediff \right>}{\expend[\buyer](\price[][\iter], 1)} \label{eq:neg_exp_pos_def}
\end{align}
Finally, we note that $\grad[\price]^2 \expend[\buyer]$ is negative semi-definite, meaning that we have $\left< \grad[\price]^2 \expend[\buyer](\price[][\iter] + c \pricediff, 1) \pricediff, \pricediff \right> \leq 0$, allowing us to re-express \Cref{eq:neg_exp_pos_def} as follows:
\begin{align}
   &= \left(\frac{4}{3}+\frac{20\elastic}{27}\right) \sum_{l \in \goods} \frac{\marshallian[\buyer][l][\iter]}{\price[l][\iter]}(\pricediff[l])^2  + \frac{25\elastic^2 }{324} \sum_{l \in \goods} \frac{(\pricediff[l])^2}{\price[l][\iter]}\marshallian[\buyer][l](\price[][\iter] + c \pricediff, \budget[\buyer]) - \notag \\
     & \left< \marshallian[\buyer](\price[][\iter], 1), \pricediff \right> + \left|\frac{\nicefrac{\budget[\buyer]}{2} \left< \grad[\price]^2 \expend[\buyer](\price[][\iter] + c \pricediff, 1) \pricediff, \pricediff \right>}{\expend[\buyer](\price[][\iter], 1)} \right|
\end{align}


\begin{align}
    &\leq \left(\frac{4}{3}+\frac{20\elastic}{27}\right) \sum_{l \in \goods} \frac{\marshallian[\buyer][l][\iter]}{\price[l][\iter]}(\pricediff[l])^2  + \frac{25\elastic^2 }{324} \sum_{l \in \goods} \frac{(\pricediff[l])^2}{\price[l][\iter]}\marshallian[\buyer][l](\price[][\iter] + c \pricediff, \budget[\buyer]) - \left< \marshallian[\buyer](\price[][\iter], 1), \pricediff \right> + \notag \\
     &\frac{5\elastic }{6}    \sum_{l \in \goods} \frac{(\pricediff[l])^2}{\price[l]}  \marshallian[\buyer][l](\price[][\iter] + c \pricediff, \budget[\buyer])\\
    &\leq \left(\frac{4}{3}+\frac{20\elastic}{27}\right) \sum_{l \in \goods} \frac{\marshallian[\buyer][l][\iter]}{\price[l][\iter]}(\pricediff[l])^2  + \notag \\
     & \left(\frac{5 \elastic}{6} + \frac{25\elastic^2 }{324} \right) \sum_{l \in \goods} \frac{(\pricediff[l])^2}{\price[l][\iter]}\marshallian[\buyer][l](\price[][\iter] + c \pricediff, \budget[\buyer]) - \left< \marshallian[\buyer](\price[][\iter], 1), \pricediff \right>
\end{align}

where the penultimate line was obtained from \Cref{lemma:expend_change}.
\end{proof}